\documentclass[11pt]{article}
\usepackage{amsmath}
\usepackage{amssymb}
\usepackage{amsthm}
\usepackage{array}
\usepackage{appendix}
\usepackage[english]{babel}
\usepackage{bbm}
\usepackage{color}
\usepackage{comment}
\usepackage{geometry} 
\usepackage{graphicx}
\usepackage{hyperref}
\usepackage[utf8]{inputenc}
\usepackage[T1]{fontenc}
\usepackage{lmodern}
\usepackage{microtype}
\usepackage{booktabs}
\usepackage{longtable}
\usepackage{lscape}
\usepackage{natbib}
\setcitestyle{authoryear,open={(},close={)}} 
\usepackage{pdflscape}
\usepackage{siunitx}
\usepackage{tikz}
\usepackage{latexsym}
\usepackage{dsfont}
\usepackage[font=small,labelfont=bf]{caption}
\usepackage{paralist}
\tikzstyle{solid node}=[circle,draw,inner sep=2,fill=black] 
\tikzstyle{hollow node}=[circle,draw,inner sep=2.5] 
\tikzstyle{hollow node1}=[rectangle,draw,inner sep=2] 
\theoremstyle{plain}
\newtheorem{thm}{Theorem}[section]
\newtheorem{prop}[thm]{Proposition}
\newtheorem{assum}[thm]{Assumption}

\newtheorem{cor}[thm]{Corollary}
\newtheorem{definition}[thm]{Definition}

\theoremstyle{definition}
\newtheorem{example}[thm]{Example}

\theoremstyle{remark}
\newtheorem{rmk}[thm]{Remark}

\graphicspath{ {pictures/} }

\renewcommand{\P}{\operatorname{\mathbb{P}}}
\newcommand{\E}{\operatorname{\mathbb{E}}}
\newcommand{\G}{\mathcal{G}}

\newcommand{\cov}{\operatorname{cov}}
\newcommand{\var}{\operatorname{var}}
\newcommand{\pth}[2]{({#1} \rightsquigarrow {#2})}
\newcommand{\indep}{\perp\!\!\!\perp}

\newcommand{\reals}{\mathbbm{R}}
\newcommand{\ind}{\mathbbm{1}}

\newcommand{\diff}{\mathrm{d}}
\newcommand{\cd}{\mathrm{cd}}

\newcommand{\dto}{\stackrel{d}{\longrightarrow}}
\newcommand{\pnorm}{\operatorname{\Phi}}

\newcommand{\1}{\mathds{1}}

\providecommand{\keywords}[1]
{
	{\small	
	\textbf{\textit{Keywords ---}} #1}
}
\definecolor{gray25}{rgb}{0.25, 0.25, 0.25}
\definecolor{gray50}{rgb}{0.50, 0.50, 0.50}
\definecolor{navyblue}{rgb}{0.0, 0.0, 0.5}
\newcommand{\js}[1]{\textcolor{navyblue}{\sffamily\footnotesize [#1]}}
\newcommand{\sa}[1]{\textcolor{purple}{\sffamily\footnotesize [#1]}}

\definecolor{mygreen}{rgb}{0.10, 0.6,0 }
\definecolor{violet}{rgb}{0.6, 0, 0.2}

\definecolor{ggr}{HTML}{006633}
\definecolor{bbl}{HTML}{0066CC}
\definecolor{ppn}{HTML}{CC0099}
\definecolor{darkgreen}{rgb}{0, .5, 0}

\title{Extremes of Markov random fields on block graphs: max-stable limits and structured H\"{u}sler--Reiss distributions}
\author{Stefka Asenova\thanks{Corresponding author. UCLouvain, LIDAM/ISBA, Voie du Roman Pays 20, 1348 Louvain-la-Neuve, Belgium. E-mail: stefka.asenova@uclouvain.be} \and Johan Segers\thanks{UCLouvain, LIDAM/ISBA, Voie du Roman Pays 20, 1348 Louvain-la-Neuve, Belgium. E-mail: johan.segers@uclouvain.be}}
\date{\today}
\begin{document}

\maketitle

\begin{abstract}
	We study the joint occurrence of large values of a Markov random field or undirected graphical model associated to a block graph.
	On such graphs, containing trees as special cases, we aim to generalize recent results for extremes of Markov trees. 
	Every pair of nodes in a block graph is connected by a unique shortest path. These paths are shown to determine the limiting distribution of the properly rescaled random field given that a fixed variable exceeds a high threshold. 
	The latter limit relation implies that the random field is multivariate regularly varying and it determines the max-stable distribution to which component-wise maxima of independent random samples from the field are attracted.
	
	When the sub-vectors induced by the blocks have certain limits parametrized by Hüsler--Reiss distributions, the global Markov property of the original field induces a particular structure on the parameter matrix of the limiting max-stable Hüsler--Reiss distribution.
	The multivariate Pareto version of the latter turns out to be an extremal graphical model according to the original block graph.
	Thanks to these algebraic relations, the parameters are still identifiable even if some variables are latent.
\end{abstract}

\keywords{Markov random field; graphical model; block graph; multivariate extremes; tail dependence; latent variable; Hüsler--Reiss distribution; conditional independence}

\section{Introduction}







Graphical models are statistical models for random vectors whose components are associated to the nodes of a graph, and edges serve to encode conditional independence relations.
They bring structure to the web of dependence relations between the variables.
In the context of extreme value analysis, they permit for instance to model the joint behavior of the whole random vector given that a specific component exceeds a high threshold.
This could concern a system of intertwined financial risks, one of which is exposed to a large shock, or measurements of water heights along a river network, when a high water level is known to have occurred at a given location.
Each time, the question is how such an alarming event affects conditional probabilities of similarly high values occurring elsewhere.

A Markov random field is a random vector satisfying a set of conditional independence relations with respect to a non-directed graph.
For a max-stable random vector with continuous and positive density, \citet{strokorb2016conditional} showed that conditional independence implies unconditional independence. 
This implies that an absolutely continuous max-stable distribution can satisfy the Markov property with respect to a non-trivial graph only if variables on non-adjacent nodes are actually independent.
These models clearly differ from the max-linear graphical models in \citet{gissibl2018max} and \citet{amendola2022conditional} with respect to directed acyclic graphs which have max-stable, but singular, distributions. In our paper, conditional independence relations are induced by separation properties in undirected graphs, not by parent-child relations in directed ones.

Markov random fields with continuous and positive densities have made their way in extreme value analysis through the lens of multivariate Pareto distributions. 
Multivariate generalized Pareto distributions arise as weak limits of the normalized excesses over a threshold given the event that at least one variable exceeds a high threshold \citep{rootzen2006multivariate}. For a random vector $Y=(Y_1, \ldots, Y_d)$ with a multivariate Pareto distribution and a positive, continuous density, \citet{engelke2020graphical} study conditional independence for the vectors $Y^{(k)} = (Y \mid Y_k>1)$. They define $Y$ as an extremal graphical model with respect to the graph for which $Y^{(k)}$ is an ordinary graphical model. In our paper we provide an example of a probabilistic graphical model in the classical sense whose graph is shared with the corresponding extremal graphical model. The definition of extremal graphical model in \citet{engelke2020graphical} is made on the basis of the fact that $Y^{(k)}$ has support on a product space. More general results defining conditional independence on the punctured $d$-dimensional Euclidean space $\mathbb{R}^d\setminus  \{0\}$ is given in \citet{strokorb2022graphical}.

We focus on Markov random fields with respect to connected block graphs, which generalize trees, because one obtains a block graph if the edges of a tree are replaced by complete subgraphs.
Block graphs share some key properties with trees, such as unique shortest paths, acyclicality outside cliques and unique minimal separators. 
We study the limiting behavior of the normalized random field when a given variable exceeds a high threshold and we show that the limit depends on the unique shortest paths, a result familiar from \citet{segers2020one}. 
As a prime example, we consider the case where the random vectors induced by the graph's cliques have limits determined by Hüsler--Reiss distributions.

Our main result, Theorem~\ref{prop_main}, is inspired by the one about Markov random fields on connected trees, or Markov trees in short, in \citet{segers2020one}. Theorem~1 therein states that the limiting distribution of the scaled Markov tree given that a high threshold is exceeded at a particular node is a vector composed of products of independent multiplicative increments along the edges of the unique paths between the nodes. Here we show a similar result for block graphs rather than trees. This time, the products are with respect to the unique \emph{shortest} path between pairs of nodes. The increments over the edges are independent between blocks but possibly dependent within blocks. 
The product structure of the limiting field originates from the asymptotic theory for Markov chains with a high initial state, going back to \citet{smith1992the}, \citet{perfekt1994extremal}, and \citet{yun1998the}. It is confirmed by many subsequent studies on Markov chains such as \citet{segers2007multivariate}, \citet{janssen2014markov}, \citet{papastat2017extreme} and \citet{papastat2019hidden}. 

The assumptions that underlie Theorem~\ref{prop_main} are rather common in the literature. 
The first one 
says that the distribution of the random vector indexed by nodes within a block and scaled by the value of one of the variables, conditional on that variable taking a high value, converges to a non-degenerate distribution. 
The assumption is similar to the one in \citet{heffernan2004aconditional} to model tail probabilities in case of asymptotic independence. Our version of the assumption implies multivariate regular variation of the clique vectors \citep[Theorem~2]{segers2020one}, making our model suitable for asymptotically dependent random vectors. An in-depth analysis of more refined forms of regular variation designed for modelling joint tails in case of asymptotic independence is given for instance in \citet{resnick2005hidden,resnick2007limit}, see also \citet{resnick2002hidden} and \citet{campos2005extremal}. The topic of asymptotic independence is investigated in \citet{papastat2017extreme} and \citet{papastat2019hidden} for possibly higher-order Markov chains and in \citet{strokorb2020extremal} for extremal graphical models.

  The second assumption in our Theorem~\ref{prop_main} 
excludes processes which can become extreme again after reaching non-extreme levels. Earlier literature based on such regular behavior is \citet{smith1992the}, \citet{perfekt1994extremal}, \citet{segers2007multivariate}, \citet{janssen2014markov}, and \citet{resnick2013asymptotics}. A toy model violating this principle is given in \citet[Example~7.5]{segers2007multivariate}. Markov chains exhibiting transitions from non-extreme to extreme regions are studied extensively in \citet{papastat2017extreme} and \citet{papastat2019hidden}.

While the generalization of Markov trees to Markov fields on a larger class of graphs in Theorem~\ref{prop_main} is the first novelty of our paper, the second novelty concerns the domain of attraction of the Hüsler--Reiss distribution in Section~\ref{sec:hr}.
%
%
For the study of extremal graphical models, the Hüsler--Reiss distribution offers many advantages akin to those of the Gaussian distribution for ordinary graphical models \citep{engelke2020graphical}.
In Section~\ref{sec:hr}, we study the implications of our main result for a Markov random field with respect to a block graph which has \emph{clique-wise limits} based on Hüsler--Reiss distributions, or \emph{Markov block graph with Hüsler--Reiss limits} in short. 
In Proposition~\ref{prop:log-excess}, we show that for the said Markov block graphs with Hüsler--Reiss limits, the limiting distributions in Theorem~\ref{prop_main} are all multivariate log-normal.
In \citet{engelke2014estimation}, such log-normal limits were found to characterize the domain of attraction of the Hüsler--Reiss max-stable distribution.
Proposition~\ref{prop_HRattractor} states that the parameter matrix of the max-stable limit has an explicit and elegant path-sum structure, as was found for Markov trees in \citet{segers2020one} and \citet{asenova2021inference}. 
By Proposition~\ref{prop:extr_gm}, the associated multivariate Pareto distribution is an extremal graphical model with respect to the same graph, complementing Proposition~4 in \cite{engelke2020graphical}.
Finally, it is proved in Proposition~\ref{prop:identif_hr} that all parameters of the limiting structured Hüsler--Reiss max-stable distribution remain identifiable even when some of the variables are latent, as long as these variables lie on nodes belonging to least three different cliques.
This generalizes a similar identifiability claim for trees in \citet{asenova2021inference}.
Our contribution can be cast within the classical paradigm in extreme value analysis to model maxima and high threshold excesses by distribution families validated by asymptotic theory.
We find that probabilistic graphical models with respect to block graphs are in the domain of attraction of max-stable distributions that enjoy a particular structure induced by the graph and which for the Hüsler--Reiss family yield multivariate Pareto distributions that are extremal graphical models as in \citet{engelke2020graphical}.
Therefore, we believe that our results provide additional justification for the practical use of such models in situations where the data-generating mechanism can be described (approximately) as a graphical model on a block graph.
The point is important when choosing between models. Indeed, starting from component-wise maxima of Gaussian vectors with structured correlation matrices, \citet{lee2017multivariate} propose a different way of incorporating graphical or factor structures into Hüsler--Reiss max-stable distributions; see \citet*[Section~A.2]{asenova2021inference} for a discussion.

In our perspective the graph is known and we are interested in the tail limits of a random vector living on the nodes on that graph. This is suitable in applications such as extremes on river networks, or any application where we agree a priori on some graph. In \citet{engelke2020structure} the graph structure describing the dependence of extremes is the object of interest and the method proposed is nonparametric. A graph structure discovery based on Hüsler--Reiss Pareto models is proposed in \citet{engelke2020graphical}.

The outline of the paper is as follows. In the preliminary Section~\ref{sec:prelim} we introduce concepts and notation from graph theory, graphical models and extreme value analysis.
The main result about the convergence of the rescaled random field, conditional on the event that a given variable exceeds a high threshold, is stated in Section~\ref{sec:eot}.
Section~\ref{sec:hr} concerns a Markov block graph composed of clique-wise  distributions whose limits are determined by the Hüsler--Reiss family.
The conclusion in Section~\ref{sec:concl} summarizes the main points and sketches some directions for further research. Proofs are deferred to  Appendices~\ref{app:eot} and~\ref{app:hr}.

\section{Preliminaries}
\label{sec:prelim}

\subsection{Graph theory and Markov random fields}

A graph $\G = (V, E)$ is a pair consisting of a finite, non-empty vertex (node) set $V$ and edge set $E \subseteq \{(a,b)\in V\times V: a\neq b\}$. 
Often, we will write $e \in E$ for a generic edge. 
The graph $\G$ is said to be non-directed if for every pair of nodes $a, b$ we have $(a,b)\in E$ if and only if $(b, a)\in E$.
A path from node $a$ to node $b$ is an ordered sequence of vertices $(v_1, \ldots, v_n)$ with $v_1 = a$ and $v_n = b$ such that $(v_i, v_{i+1}) \in E$ for all $i = 1,\ldots,n-1$ and in which all nodes are distinct, except possibly for the first and last nodes.
A cycle is a path from a node to itself.
Two distinct nodes are connected if there exists a path from one node to the other.
A graph is connected if every pair of distinct nodes is connected.
In this paper, we only consider connected, undirected graphs.
 
An induced subgraph $\G_A=(A,E_A)$ is formed from the vertices in a subset $A$ of $V$ and all edges connecting them, $E_A=\{(a,b)\in E: a, b\in A\}$. A graph is complete if every pair of distinct nodes is an edge. A set of nodes $C \subseteq V$ is said to be a clique if the induced subgraph $\G_C=(C,E_C)$ is complete; the latter graph will be called a clique as well. A clique is maximal if it is not properly contained in larger one. Further on we use the word `clique' to mean maximal clique. 
The set of all (maximal) cliques of $\G$ will be denoted by $\mathcal{C}$.
 
A separator set $S \subseteq V$ between two other vertex subsets $A$ and $B$ is such that every path from a node in $A$ to a node in $B$ passes through at least one node in $S$.  A separator set $S$ is minimal when there is no proper subset of $S$ which is a separator of $A$ and $B$ too. 

In this paper we consider connected block graphs. A block is a maximal biconnected component, i.e., a subgraph that will remain connected after the removal of a single node. A block graph is a graph where every block is a clique; see Figure~\ref{fig:cg7} for an example.
If the edge between nodes $2$ and $6$ were removed, the subgraph induced by $\{2, 4, 5, 6\}$ would still be a block, i.e., biconnected, but it would no longer be a clique, and so the graph would no longer be a block graph.
Block graphs are considered natural generalizations of trees \citep{le2010the}. 

A path between two (distinct) nodes $a$ and $b$ is said to be shortest if no other path between $a$ and $b$ contains less nodes.
In block graphs, any two nodes $a$ and $b$ are connected by a unique shortest path \citep[Theorem~1]{behtoei2010a}, i.e., any other path connecting $a$ and $b$ contains strictly more nodes than the given path.
If the shortest path between $a$ and $b$ is the ordered node set $(v_1, \ldots, v_n)$, with $v_1 = a$ and $v_n = b$, then we define $\pth{a}{b}$ as the set of edges
\[
	\pth{a}{b}=\{(v_1, v_2), \ldots, (v_{n-1}, v_n)\}.
\]
Another important property of block graphs is that cycles can only occur within cliques, i.e., a path which is not contained in a single clique has two different endpoints. Moreover, in a block graph, a minimal separator between two cliques is always a single node. Two distinct cliques have at most one node in common. The set of minimal clique separators will be denoted by $\mathcal{S}$. In the block graph in Figure~\ref{fig:cg7}, the collection of minimal clique separators is $\mathcal{S} = \{ \{2\}, \{6\} \}$.

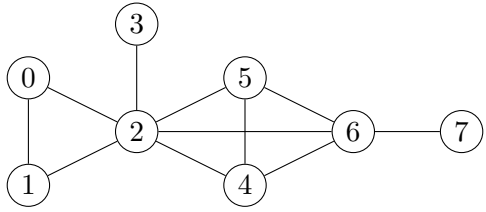
\begin{figure}
	\begin{center}
	\begin{tikzpicture}[scale=0.95] 
		\node[hollow node] (0) at (0,0) {0};
		\node[hollow node] (1) at (0,-1.5)  {1};
		\node[hollow node] (2) at (1.5,-0.75)  {2};
		\node[hollow node] (3) at (1.5,0.75)  {3};
		\node[hollow node] (6) at (4.5,-0.75)  {6};
		\node[hollow node] (4) at (3,-1.5)  {4};
		\node[hollow node] (5) at (3,0)  {5};
		\node[hollow node] (7) at (6,-0.75)  {7};
		\path (0) edge (1)  ;
		\path (1) edge (2);
		\path (2) edge (0)  ;
		\path (2) edge (4);
		\path (6) edge (2);
		\path (2) edge (5);
		\path (4) edge (5);
		\path (4) edge (6);
		\path (3) edge (2);
		\path (6) edge (7);
		\path (6) edge (5);
	\end{tikzpicture}
	\end{center}
	\caption{
		 An example of a block graph. There are four blocks or (maximal) cliques, $\mathcal{C} = \bigl\{ \{0, 1, 2\}, \{2, 3\}, \{2, 4, 5, 6\}, \{6, 7\} \bigr\}$, as well as two minimal separators, $\mathcal{S} = \bigl\{ \{2\}, \{6\} \bigr\}$. The unique shortest path from $7$ to $0$ has edge set $\pth{7}{0} = \{(7, 6), (6, 2), (2, 0)\}$.}
	\label{fig:cg7}
\end{figure}

Let $X = (X_v, v \in V)$ be a random vector which is indexed by the node set, $V$, of a graph $\G = (V, E)$. For non-empty $W \subseteq V$, write $X_W = (X_v, v \in W)$. We say that $X$ is a Markov random field with respect to $\G$ if it satisfies the global Markov property, that is, for all non-empty disjoint node sets $A, B, S \subset V$ we have the implication
\[
	\text{$S$ is a separator of $A$ and $B$ in $\G$} \implies
	X_A \indep X_B \mid X_S,
\]
where the right-hand side means that $X_A$ and $X_B$ are independent conditionally on $X_S$.
In other words, conditional independence relations within $X$ are implied by separation properties in $\G$. An extensive treatment of conditional independence, Markov properties and graphical models can be found in \citet{lauritzen1996graphical}.




Often we will use a double subscript to a random vector, e.g., $X_{u,A}$ for $A\subseteq V$, which, if not indicated otherwise, will mean that it is a vector indexed by the elements of $A$ and in some way related to a particular node $u$, not necessarily in $A$. For scalars, the expressions $x_{u,v}$ and $x_{uv}$ will signify the same thing and if $e=(u,v)$ is an edge they can also be written as $x_e$. In case of iterated subscripts, we will prefer the comma notation, $x_{u_1, u_2}$.

For two non-empty sets $A$ and $B$, let $B^A$ denote the set of functions $x : A \to B$. Formally, we think of $x$ as a vector indexed by $A$ and with elements in $B$, as reflected in the notation $x = (x_a, a \in A)$. We will apply this convention most often to subsets $A$ of the node set $V$ of a graph $\G$ and to subsets $B$ of the extended real line.


\subsection{Max-stable and multivariate Pareto distributions}
\label{sec:evt-intro}

Let $X = (X_v, v \in V)$ be a random vector indexed by a finite, non-empty set $V$ and with joint cumulative distribution function $F$, the margins of which are continuous, i.e., have no atoms. The interest in this paper is in tail dependence properties of $X$. It is convenient and does not entail a large loss of generality to  assume that the margins have been standardized to a common distribution, a convenient choice of which will be either the unit-Pareto distribution, $\P(X_v > x) = 1/x$ for $v \in V$ and $x \ge 1$, or the unit-Fréchet distribution, $\P(X_v \le x) = \exp(-1/x)$ for $v \in V$ and $x > 0$.
Assume that $F$ is in the max-domain attraction of a multivariate extreme-value distribution $G$, i.e., for either of the two choices of the common marginal distribution we have
\begin{equation}
\label{eq:FDG}
	\forall z \in (0, \infty)^V, \qquad \lim_{n \to \infty} F^n(nz) = G(z),
\end{equation}
a condition which will be denoted by $F \in D(G)$. Let $X^{(n)} = (X_v^{(n)}, v \in V)$ for $n = 1, 2, \ldots$ be a sequence of independent and identically distributed random vectors with common distribution $F$. Let $M^{(n)} = \bigl( M_v^{(n)}, v \in V \bigr)$ with $M_v^{(n)} = \max_{i=1,\ldots,n} X_v^{(i)}$ be the vector of component-wise sample maxima. Equation~\eqref{eq:FDG} then means that
\begin{equation}
\label{eq:FDGMn}
	M^{(n)} / n \dto G, \qquad n \to \infty,
\end{equation}
the arrow $\dto$ signifying convergence in distribution.
The choice of the scaling sequence $n$ in~\eqref{eq:FDG} and~\eqref{eq:FDGMn} is dictated by the marginal standardization and implies that the margins of $G$ are unit-Fréchet too, i.e., $G$ is a simple max-stable distribution. 
The latter can be written as
\begin{equation*} 
	G(z)=
	\exp\left(-\mu\left[\left\{x \in [0, \infty)^V : \ \exists v \in V, \, x_v>z_v \right\}\right]\right),
	\qquad z \in (0, \infty]^V,
\end{equation*}
where the exponent measure $\mu$ is a non-negative Borel measure on the punctured orthant $[0, \infty)^V \setminus \{0\}$, finite on subsets bounded away from the origin \citep{haan1977limit, resnick1987extreme}. The function $\ell : [0, \infty)^d \to [0, \infty)$ is defined by
\begin{equation} \label{eqn:stdf}
	\ell(y)
	:= - \ln G(1/y_v, v \in V)
	= \mu\left[\left\{x \in [0, \infty)^V : \ \exists v \in V, \, x_v > 1/y_v \right\}\right].
\end{equation}
More background on multivariate extreme value analysis can be found for instance in the monographs \citet{resnick1987extreme}, \citet{beirlant2004statistics} and \citet{haan2007extreme}.

We can replace the integer $n$ in~\eqref{eq:FDG} by the real scalar $t > 0$: the condition $F \in D(G)$ is equivalent to
\begin{equation} \label{eqn:FtG}
	\forall z \in (0, \infty)^V, \qquad \lim_{t \to \infty} F^t(tz) = G(z).
\end{equation}
By a direct calculation starting from~\eqref{eqn:FtG} it follows that
\begin{equation}
	\label{eqn:pareto_lim}
	\lim_{t\to\infty} \P \left( 
		\forall v \in V, \, X_v / t \leq z_v \mathrel{\Big\vert}
		\max_{v \in V} X_v > t 
	\right)
	=
	\frac{\ln G\bigl(\min(z_v, 1), v\in V\bigr)-\ln G(z)}{\ln G(1, \ldots, 1)},
\end{equation}
for $z \in (0, \infty)^V$, from which we deduce the weak convergence of conditional distributions
\begin{equation} \label{eqn:pareto}
	\left( t^{-1} X \mathrel{\Big\vert} \max_{v \in V} X_v > t \right)
	\dto 
	Y, \qquad t \to \infty,
\end{equation}
where $Y = (Y_v, v \in V)$ is a random vector whose cumulative distribution function is equal to the right-hand side in~\eqref{eqn:pareto_lim}. The law of $Y$ is a multivariate Pareto distribution and has support contained $[0, \infty)^V \setminus [0, 1]^V$. Upon a change in location, it is a member of the family of multivariate generalized Pareto distributions. The latter arise in \citet{rootzen2006multivariate} and \citet[Section~8.3]{beirlant2004statistics} as limit laws of multivariate peaks over thresholds; see also \citet{rootzen2018multivariate}.

\section{Tails of Markov random fields on block graphs} \label{sec:eot}

Let $X = (X_v, v \in V)$ be a non-negative Markov random field with respect to the connected block graph $\G = (V, E)$, or Markov block graph in short.
Suppose that at a given node $u \in V$ the variable $X_u$ exceeds a high threshold, say $t$.
This event can be expected to affect conditional probabilities of the other variables $X_v$ too.
Our main result, Theorem~\ref{prop_main}, states that, starting out from $u$, every other variable $X_v$ feels the impact of the shock at $u$ through a multiplication of increments on the edges forming the unique shortest path from $u$ to $v$. The increments are independent between cliques but possibly dependent within cliques.

After discussing the set-up and the assumptions in Section~\ref{subsec:set-up}, we state and illustrate the main result in Section~\ref{subsec:main}. Consequences for multivariate regular variation and max-domains of attraction are treated in Section~\ref{subsec:mrv}, followed by a focus in Section~\ref{subsec:max-stable} on the special case where clique vectors are max-stable already.

\subsection{Set-up}
\label{subsec:set-up}

We will be making two assumptions on the conditional distribution of $X$ at high levels.
Assumption~\ref{ass:nu} is the main one, as it will determine the limit distribution in Theorem~\ref{prop_main} through the construction in Definition~\ref{def:increments} below.
For a set $A$ and an element $b \in A$, we write $A \setminus b$ rather than $A \setminus \{b\}$.

\begin{assum}{} 
	\label{ass:nu}
	For every clique $C\in \mathcal{C}$ and every node $s\in C$ there exists a probability distribution $\nu_{C,s}$ on $[0, \infty)^{C\setminus s}$ such that, as $t\rightarrow\infty$, we have
	\[
		\mathcal{L}\left(\frac{X_v}{t}, \, v\in {C\setminus s} \mathrel{\Big|} X_{s} = t \right)
		\dto \nu_{C,s}.
	\]
\end{assum}

In the special case that the distribution of the clique vector $X_C$ is max-stable, the limit $\nu_{C,s}$ can be calculated by means of Proposition~\ref{prop:nustdf} below.


\begin{definition}[Increments]
\label{def:increments}
Under Assumption~\ref{ass:nu}, define, for fixed $u \in V$, the following $(|V|-1)$-dimensional non-negative random vector $Z$:
\begin{enumerate}[({Z}1)]
	\item
	For each clique $C$, let $s$ be the separator node in $C$ between $u$ and the nodes in $C$. If $u \in C$, then simply $s = u$, whereas if $u \not\in C$, then $s$ is the unique node in $C$ such that for every $v \in C$, any path from $u$ to $v$ passes through $s$. Note that, for fixed $u$, the node $s$ is a function of $C$, but we will suppress this dependence from the notation.
	\item 
	For each $C$, consider the limit distribution $\nu_{C,s}$ of Assumption~\ref{ass:nu}, with separator node $s \in C$ determined as in~(Z1).
	\item 
	Put $Z := (Z_{s, C\setminus s}, C \in \mathcal{C})$ where, for each $C \in \mathcal{C}$, the random vector $Z_{s, C\setminus s} = (Z_{sv}, v \in C \setminus s)$ has law $\nu_{C, s}$ and where these $|\mathcal{C}|$ vectors are mutually independent as $C$ varies.
\end{enumerate}
\end{definition}

The distribution of the random vector $Z$ in Definition~\ref{def:increments} depends on the source node $u$, but this dependence is suppressed in the notation.
The dimension of $Z$ is indeed equal to $|V| - 1$, since the sets $C \setminus s$ form a partition of $V \setminus u$ as $C$ varies in $\mathcal{C}$.
In fact, in the double index in $Z_{sv}$, every node $v \in V \setminus u$ appears exactly once; the node $s$ is the one just before $v$ itself on the shortest path from $u$ to $v$.

\begin{example}  
	\label{ex1} 
	Consider a Markov random field on the block graph in Figure~\ref{fig:cg7}. Suppose that the variable exceeding a high threshold is the one at node $u = 7$. For paths departing at~$u$, the separator nodes associated to the four cliques are as follows:
	\[
		\begin{array}{lcl}
		\toprule
		\text{clique $C$} & \text{separator node $s \in C$} & \text{node set $C \setminus s$} \\
		\midrule
		\{0, 1, 2\} & 2 & \{0, 1\} \\
		\{2, 3\} & 2 & \{3\}\\
		\{2, 4, 5, 6\} & 6 & \{2, 4, 5\}\\
		\{6, 7\} & 7 & \{6\} \\
		\bottomrule
		\end{array}
	\]
Note that the union over the sets $C \setminus s$ is equal to $\{0, 1, \ldots, 6\} = V \setminus u$.
Assumption~\ref{ass:nu} requires certain joint conditional distributions to converge weakly: as $t\rightarrow\infty$, we have
\begin{align*} 
	\left(\frac{X_0}{t}, \frac{X_1}{t}
	\mathrel{\Big|} X_2=t \right)
	& \dto \nu_{\{0, 1, 2\}, 2} 
	, 
	&
	\left( \frac{X_3}{t} \mathrel{\Big|} X_2=t \right)
	&\dto \nu_{\{2, 3\}, 2} 
	,\\
	\left(\frac{X_2}{t}, \frac{X_4}{t}, \frac{X_5}{t} \mathrel{\Big|}X_6=t \right)
	&\dto \nu_{\{2,4,5,6\}, 6} 
	, &
	\left( \frac{X_6}{t} \mathrel{\Big|} X_7=t \right)
	&\dto \nu_{\{6, 7\}, 7} 
	.
\end{align*}
The random vector $Z$ in Definition~\ref{def:increments}, step~(Z3), is a $7$-dimensional random vector whose joint distribution is equal to the product of the above four distributions:
\begin{equation}
\label{eq:Z:example}
Z:= \left(
\textcolor{orange}{Z_{20}, Z_{21}}; 
\textcolor{purple}{Z_{23}}; 
\textcolor{blue}{Z_{62}, Z_{64}, Z_{65}}; 
\textcolor{darkgreen}{Z_{76}}
\right)
\sim \nu_{\{0,1,2\},2} \otimes \nu_{\{2, 3\}, 2} \otimes \nu_{\{2, 4, 5, 6\}, 6} \otimes \nu_{\{6, 7\}, 7}.
\end{equation}
We think of the random variable $Z_{sv}$ as being associated to the edge $(s, v) \in E$:
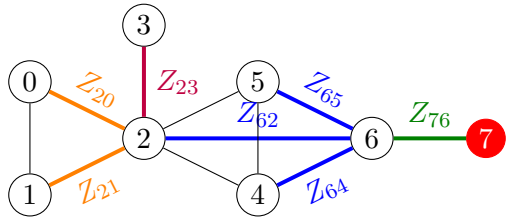
\begin{center}	
\begin{tikzpicture} 
	\node[hollow node] (0) at (0,0) {0};
	\node[hollow node] (1) at (0,-1.5)  {1};
	\node[hollow node] (2) at (1.5,-0.75)  {2};
	\node[hollow node] (3) at (1.5,0.75)  {3};
	\node[hollow node] (6) at (4.5,-0.75)  {6};
	\node[hollow node] (4) at (3,-1.5)  {4};
	\node[hollow node] (5) at (3,0)  {5};
	\node[solid node, color=red] (7) at (6,-0.75)  {\textcolor{white}{7}};
	\path[] (0) edge (1)  ;
	\path[color=orange, line width=1.5pt, sloped, anchor=north] (1) edge node{$Z_{21}$} (2);
	\path[color=orange, line width=1.5pt, sloped, anchor=south] (2) edge node{$Z_{20}$} (0)  ;
	\path[] (2) edge (4);
	\path[color=blue, line width=1.5pt, anchor=south] (6) edge node{$Z_{62}$} (2);
	\path[] (2) edge (5);
	\path[] (4) edge (5);
	\path[color=blue, line width=1.5pt, anchor=north, sloped] (4) edge node{$Z_{64}$} (6);
	\path[color=purple, line width=1.5pt, anchor=west] (3) edge node{$Z_{23}$} (2);
	\path[color=darkgreen, line width=1.5pt, anchor=south] (6) edge node{$Z_{76}$ } (7);
	\path[color=blue, line width=1.5pt, anchor=south, sloped] (6) edge node{$Z_{65}$} (5);
\end{tikzpicture}
\end{center}	
By construction, the random sub-vectors $(Z_{20}, Z_{21})$, $Z_{23}$, $(Z_{62}, Z_{64}, Z_{65})$ and $Z_{76}$ are independent from each other and their marginal distributions are $(Z_{20}, Z_{21}) \sim \nu_{\{0, 1, 2\}, 2}$ and so on.
Every node $v \in \{0, 1, \ldots, 6\}$ appears exactly once as a second index of a variable in $Z_{sv}$ in~\eqref{eq:Z:example}.
For each such $v$, the first index $s$ is the node right before $v$ on the path from $u=7$ to $v$.

In the same block graph, we could also suppose that the variable exceeding a high threshold is the one on node $u = 4$. The picture would then change as follows:
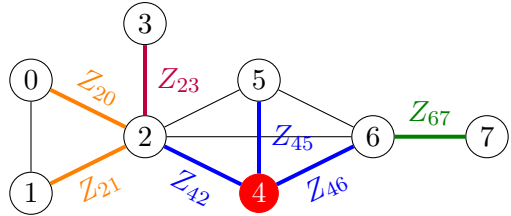
\begin{center}
\begin{tikzpicture} 
	\node[hollow node] (0) at (0,0) {0};
	\node[hollow node] (1) at (0,-1.5)  {1};
	\node[hollow node] (2) at (1.5,-0.75)  {2};
	\node[hollow node] (3) at (1.5,0.75)  {3};
	\node[hollow node] (6) at (4.5,-0.75)  {6};
	\node[solid node, color=red] (4) at (3,-1.5)  {\textcolor{white}{4}};
	\node[hollow node] (5) at (3,0)  {5};
	\node[hollow node] (7) at (6,-0.75)  {7};
	\path[] (0) edge (1)  ;
	\path[color=orange, line width=1.5pt, sloped, anchor=north] (1) edge node{$Z_{21}$} (2);
	\path[color=orange, line width=1.5pt, sloped, anchor=south] (2) edge node{$Z_{20}$} (0)  ;
	\path[color=blue, line width=1.5, sloped, anchor=north] (2) edge node{$Z_{42}$} (4);
	\path[] (6) edge (2);
	\path[] (2) edge (5);
	\path[color=blue, line width=1.5pt, anchor=west] (4) edge node{$Z_{45}$}(5);
	\path[color=blue, line width=1.5pt, anchor=north, sloped] (4) edge node{$Z_{46}$} (6);
	\path[color=purple, line width=1.5pt, anchor=west] (3) edge node{$Z_{23}$} (2);
	\path[color=darkgreen, line width=1.5pt, anchor=south] (6) edge node{$Z_{67}$ } (7);
	\path[] (6) edge (5);
\end{tikzpicture}
\end{center}
Again, colour-coded sub-vectors corresponding to cliques are mutually independent. The vectors $(Z_{21}, Z_{21})$ and $Z_{23}$ are equal in distribution to those as when the starting node was $u = 7$, but the vectors $(Z_{42}, Z_{45}, Z_{46})$ and $Z_{67}$ are new. In particular, $Z_{67}$ is not the same as $Z_{76}$, having different indices of the conditioning variable in Assumption~\ref{ass:nu}.
\hfill$\diamondsuit$
\end{example}

If the univariate margins of $X_C$ are of Pareto-type,  Assumption~\ref{ass:nu} implies that the distribution of $X_C$ is regularly varying according to \citet[Corollary~1 and Theorem~2]{segers2020one}; see also Corollary~\ref{cor:max-stable} below. In case the components of $X_C$ are asymptotically independent \citet{resnick1987extreme}, however, the limit measure $\nu_{C,s}$ is degenerate at zero. This knowledge is unhelpful for studying extremes of near independent data \citep{ledford1996statistics, ledford1997modelling, heffernan2004aconditional} because it does not allow for modeling, inference or extrapolation.
	To accommodate cases of asymptotic independence, a refinement of the assumptions would be necessary. This is not pursued in the paper. A starting point would be the studies on Markov chains in \citet{papastat2017extreme} and \citet{papastat2019hidden}.

In Assumption~\ref{ass:nu}, let $\nu_{C, s}^v$ denote the univariate marginal distribution corresponding to node $v \in C \setminus s$.
Recall that $\mathcal{S}$ denotes the set of minimal separator nodes between the cliques in the block graph.

\begin{assum}{}
	\label{ass:zero}
	
	Let $\{u, \ldots, s\}$ be the sequence of nodes of the unique shortest path between two nodes $u\in V$ and $s\in \mathcal{S}$. Let $C$ be any clique which contains $s$, but no other node of $\{u, \ldots, s\}$. If there is an edge $(a,b)$ on the path $\pth{u}{s}$ such that $\nu_{C',a}^b(\{0\}) > 0$, where $C'$ is the (unique) clique containing the nodes $a$ and $b$, then for any $\eta>0$, we have
\begin{equation} \label{eqn:zero}
	\limsup_{\delta\downarrow 0}
	\limsup_{t\rightarrow\infty}
	\sup_{x_s \in [0,\delta]}
	\P\left( 
		\exists v \in C \setminus s : X_v/t > \eta \mid X_s/t = x_s 
	\right)
	=
	0.
\end{equation}


\end{assum}


Figure~\ref{fig:cpc} illustrates the scope of Assumption~\ref{ass:zero}.  For the variables in clique $C'$, Assumption~\ref{ass:nu} implies that as $t\rightarrow \infty$,
	\[
	\mathcal{L}\left(\frac{X_v}{t}, \, v\in {C'\setminus a} \mathrel{\Big|} X_{a} = t \right)
	\dto \nu_{C',a}.
	\]
The univariate margin $\nu_{C',a}^b$ related to component $b \in C' \setminus a$ could have positive mass at zero, which is the condition $\nu_{C',a}^b(\{0\})>0$ in Assumption~\ref{ass:zero}. 
	The condition is similar to the one in \citet[page~860]{segers2020one} for Markov trees and the one for Markov chains in \citet*[Section~4]{papastat2017extreme}. 
	In Figure~\ref{fig:cpc}, consider the clique $C$ with nodes $s,v_1,v_2$. Relation~\eqref{eqn:zero} says that the probability that at least one rescaled variable, $X_{v_1}/t$ or $X_{v_2}/t$, exceeds a small but positive value given that at the separator node, $s$, the value is very small already, is close to zero. In other words,
	in the limit all scaled variables must be smaller than any arbitrary small value, $\eta>0$, given that the scaled variable $X_s/t$ at the separator is known to have a small value, i.e., $x_s\in [0,\delta]$ with $\delta\downarrow 0$. Predecessors of condition~\eqref{eqn:zero} are found in earlier literature on Markov chains: \citet[page~39]{smith1992the} and \citet[page~537]{perfekt1994extremal}. \citet{papastat2017extreme} and \citet{papastat2019hidden} provide limiting results for Markov chains of arbitrary order when the process is allowed to switch between non-extreme and extreme states. Example~7.5 in \citet{segers2007multivariate} provides a Markov chain that does not satisfy Assumption~\ref{ass:zero}.
	
\begin{figure}
\begin{center} 
	\begin{tikzpicture}
		\node[hollow node] (u) at (0,0)  {$u$};
		\node[] (ld) at (1.5,0)  {$\cdots$};
		\node[hollow node] (a) at (3,0)  {$a$};
		\node[hollow node] (b) at (4.5,0)  {$b$};
		\node[] (ld2) at (6,0)  {$\cdots$};
		\node[hollow node] (s) at (7.5,0)  {$s$};
		\node[hollow node] (v1) at (9,0.75)  {$v_1$};
		\node[hollow node] (v2) at (9,-0.75)  {$v_2$};
		\node[] (ld3) at (3.75,0.75)  {$\ldots$};
		\path[->] (u) edge (ld);
		\path[->] (ld) edge (a);
		\path[->] (a) edge (b);
		\path[->] (b) edge (ld2);
		\path[->] (ld2) edge (s);
		\path[] (s) edge (v1);
		\path[] (s) edge (v2);
		\path[] (v2) edge (v1);
		\path[] (a) edge (ld3);
		\path[] (b) edge (ld3);
		\node[] (Cp) at (3.75,1.4)  {$C'$};
		\node[] (C) at (8.5,1.4)  {$C$};
		\begin{scope}[dashed]
			\draw[] (3.75,0.2) circle (1cm);
			\draw[] (8.5,0) circle (1.2cm);
		\end{scope}
	\end{tikzpicture}
\end{center}
\caption{The path from $u$ to $s$ contains edge $(a,b)$ which is part of clique $C'$. The end node $s$ belongs to clique $C$ which contains further nodes $v_1$ and $v_2$. Assumption~\ref{ass:zero} concerns the distribution of $X_{v_j}/t$ given $X_{s}/t$ in case the $b$-th marginal of $\nu_{C',a}$ has positive mass at zero, as explained below the assumption.}
\label{fig:cpc}
\end{figure}
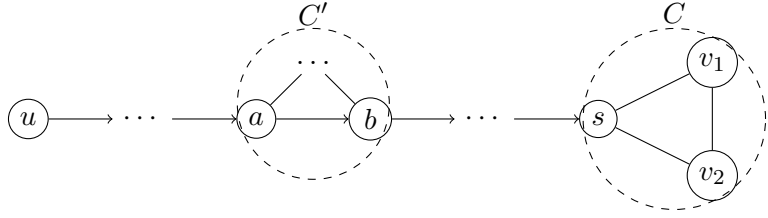

\subsection{Main result}
\label{subsec:main}

\begin{thm}
\label{prop_main}
Let $X=(X_v, v\in V)$ be a non-negative Markov random field with respect to the connected block graph $\G=(V, E)$.
Let Assumptions~\ref{ass:nu} and~\ref{ass:zero} be satisfied.
For a given $u \in V$, let $Z$ be the random vector in Definition~\ref{def:increments}.
Then as $t\rightarrow\infty$, we have
\begin{multline} 
\label{eqn:main_res}
\left( \frac{X_v}{t}, \, v\in V\setminus u  \mathrel{\Big|} X_u=t \right)
\dto (A_{uv}, v\in V\setminus u)
=:A_{u, V\setminus u} \\
\text{where} \quad A_{uv} :=\prod_{e\in \pth{u}{v}}Z_e.
\end{multline}
\end{thm}


\begin{rmk}
\label{rmk:blockvstree}
If the block graph is a tree, Theorem~\ref{prop_main} reduces to Theorem~1 in \citet{segers2020one}.
In the more general case considered here, the increments $Z_e$ can be dependent within a block, although they are still independent between blocks.
Note that for a single variable $A_{uv}$, the increments $Z_e$ appearing in~\eqref{eqn:main_res} are independent, even for a block graph that is not a tree.
The difference between a tree and a more general block graph thus manifests itself in the joint distribution of the random variables $A_{uv}$ for $v \in V \setminus u$.
\end{rmk}

\begin{example} 
	We continue with Example~\ref{ex1}. Let the variable exceeding a high threshold be the one on node $u=7$. The conclusion of Theorem~\ref{prop_main} is that as $t\rightarrow \infty$, we have
	\[
	\bigl( X_v/t, \, v\in \{0,1,\ldots,6\} \mid X_7=t \bigr)
	\dto
	\bigl( A_{7v}, \, v\in \{0,1,\ldots,6\} \bigr),
	\]
	where, in the notation of Example~\ref{ex1}, the limiting variables have the following structure:
	\begin{equation*}
	\begin{array}{lclcl}
	A_{7,\{0,1,2\}}
	&=&(A_{70},\, A_{71}, \, A_{72})
	&=&\textcolor{darkgreen}{Z_{76}}\textcolor{blue}{Z_{62}}
	(\textcolor{orange}{Z_{20}},\,
	\textcolor{orange}{Z_{20}},\, 
	1),\\[1ex]
	A_{7,\{2,3\}}
	&=&(A_{72}, A_{73})
	&=&\textcolor{darkgreen}{Z_{76}}\textcolor{blue}{Z_{62}}
	(1,\, \textcolor{purple}{Z_{23}}), \\[1ex]
	A_{7,\{2,4,5,6\}}
	&=&(A_{72},\, A_{74},\, A_{75},\, A_{76})
	&=&
	\textcolor{darkgreen}{Z_{76}}(\textcolor{blue}{Z_{62}}, 
	\textcolor{blue}{Z_{64}}, \textcolor{blue}{Z_{65}}, 1),\\[1ex]
	A_{7,\{6\}}
	&=&A_{76}
	&=&\textcolor{darkgreen}{Z_{76}}.
	\end{array}
	\end{equation*}
	The limit vector $(A_{7v})_v$ is similar to the one of a Markov field with respect to the tree formed by the unique shortest paths from node $u=7$ to the other nodes.
	The difference is that within a block, the multiplicative increments need not be independent (Remark~\ref{rmk:blockvstree}).
	\hfill$\diamondsuit$
\end{example}


\begin{rmk}
	\label{rmk:XXA}
	A useful result that we will need further on is that the convergence in~\eqref{eqn:main_res} implies a self-scaling property of the Markov random field at high levels: as $t \to \infty$, we have
	\begin{equation} \label{eq:XXA}
	\left(\frac{X_v}{X_u}, \, v\in V\setminus u \mathrel{\Big|} X_u>t \right)
	\dto (A_{u,v}, v\in V\setminus u).
	\end{equation}
	The proof of~\eqref{eq:XXA} is the same as that of Corollary~1 in \cite{segers2020one}.
\end{rmk}

\subsection{Max-domains of attraction}
\label{subsec:mrv}

In multivariate extreme value analysis, it is common to work with standardized margins, often the unit-Pareto or unit-Fr\'{e}chet distribution. Both cases are covered in the next corollary, which, in the terminology of Section~\ref{sec:evt-intro}, states that $X$ is multivariate regularly varying and in the domain of attraction of a max-stable distribution $G$ with exponent measure $\mu$ determined by the limits $A_{uv}$ in \eqref{eqn:main_res}. Let $\mathcal{M}_0$ denote the set of Borel measures $\mu$ on $[0, \infty)^V \setminus \{0\}$ such that $\mu(B)$ is finite whenever the Borel set $B$ is contained in a set of the form $\{x \in [0, \infty)^V : \max_{v \in V} x_v \ge \varepsilon\}$ for some $\varepsilon > 0$. The set $\mathcal{M}_0$ is equipped with the smallest topology that makes the evaluation mappings $\mu \mapsto \int f \, \diff \mu$ continuous, where $f$ varies over the collection of real-valued, bounded, continuous functions on $[0, \infty)^V \setminus \{0\}$ that vanish in a neighbourhood of the origin.

\begin{cor}
	\label{cor:max-stable}
	If, in addition to the assumptions in Theorem~\ref{prop_main}, we have $t \P(X_v > t) \to 1$ as $t \to \infty$ for all $v \in V$, then in $\mathcal{M}_0$ we have the convergence
	\begin{equation}
	\label{eq:tPt2nu}
		t \P(t^{-1} X \in \, \cdot \, ) \to \mu, \qquad t \to \infty,
	\end{equation}
	with limit measure $\mu$ on $[0, \infty)^V \setminus \{0\}$ determined by
	\begin{equation}
	\label{eq:intfdnu}
		\int f(x) \, \1\{x_u > 0\} \, \diff \mu(x)
		= \E \left[ \int_0^\infty f(z A_{u}) z^{-2} \, \diff z \right]
	\end{equation}
	for $u \in V$, Borel-measurable $f : [0, \infty)^V \setminus \{0\} \to [0, \infty]$, and $A_u = (A_{uv})_{v \in V}$ as in \eqref{eqn:main_res} with additionally $A_{uu} = 1$. 
	As a consequence, for $x \in (0, \infty)^V$, we have
	\begin{equation}
	\label{eq:G}
		\lim_{n \to \infty} \left[\P(X_v \le n x_v, v \in V)\right]^n
		= G(x) = \exp[- \mu(\{y : \exists v \in V, y_v > x_v\})].
	\end{equation}
\end{cor}

The indicator $\1\{x_u > 0\}$ in \eqref{eq:intfdnu} can often be omitted, simplifying the formula for $G$.

\begin{cor}
	\label{cor:max-stable:stdf}
	In Corollary~\ref{cor:max-stable}, if $u \in V$ is such that the increments $Z_e$ satisfy $\E[Z_e] = 1$ for every $e \in \pth{u}{v}$ and every $v \in V \setminus u$, then $\mu(\{x : x_u = 0\}) = 0$ and we actually have
	\begin{equation}
	\label{eq:intfdnueasy}
		\int f(x) \, \diff \mu(x)
		= \E \left[ \int_0^\infty f(z A_{u}) z^{-2} \, \diff z \right]
	\end{equation}
	for Borel-measurable $f : [0, \infty)^V \setminus \{0\}$. In particular, the stable tail dependence function of the max-stable limit $G$ in \eqref{eq:G} is then
	\begin{equation}
	\label{eq:A2stdf}
		\ell(x) = - \log G(1/x_v, v \in V) =
		\E\left[ \max \left\{ x_v A_{u,v}, v \in V \right\} \right],
		\qquad x \in [0, \infty)^V.
	\end{equation}
	For $e = (a, b)$, the condition $\E[Z_{ab}] = 1$ is equivalent to $\P(Z_{ba} > 0) = 1$.
\end{cor}

\begin{proof}[Proof of Corollary~\ref{cor:max-stable}]
	Thanks to Remark~\ref{rmk:XXA} above, the measure convergence~\eqref{eq:tPt2nu} follows from Theorem~2 in \citet{segers2020one} . By a standard argument \citep[Proposition~5.17]{resnick1987extreme}, the latter implies that $t \P(\exists v \in V, X_v > tx_v)$ converges to $\mu(\{y : \exists v \in V, y_v > x_v \})$ for $x \in (0, \infty)^V$, yielding \eqref{eq:G}.
\end{proof}

\begin{proof}[Proof of Corollary~\ref{cor:max-stable:stdf}]
	The equivalence of $\E[Z_{ab}] = 1$ and $\P(Z_{ba} > 0) = 1$ is a consequence of Corollary~3 in \citet{segers2020one}, while formula~\eqref{eq:intfdnueasy} follows from Corollary~4 in the same source and the fact that $\E[A_{uv}] = \prod_{e \in \pth{u}{v}} \E[Z_e] = 1$ for all $v \in V \setminus u$; note that the edges $e$ on a path $\pth{u}{v}$ all belong to different blocks, implying the independence of the increments $Z_e$. In combination with the identity for $G$ in~\eqref{eq:G}, setting $f$ to be the indicator function of the set $\{ y : \max_{v \in V} x_v y_v > 1\}$ in \eqref{eq:intfdnueasy} yields \eqref{eq:A2stdf} via
	\begin{align*}
		- \log G(1/x_v, v \in V) 
		= \int f \, \diff \mu 
		&= \E \left[ \int_0^\infty \1 \left\{ \max_{v \in V}  x_v A_{uv} > z^{-1} \right\} z^{-2} \, \diff z \right] \\
		&= \E\left[ \max \left\{ x_v A_{uv}, v \in V \right\} \right].
		\qedhere
	\end{align*}
\end{proof}


\subsection{Special case: max-stable clique vectors}
\label{subsec:max-stable}

The limit distribution in Theorem~\ref{prop_main} is determined by the graph structure and the clique-wise limit distributions $\nu_{C,s}$ in Assumption~\ref{ass:nu} via Definition~\ref{def:increments}. The next result provides those limits $\nu_{C,s}$ in the special case that $X_C$ follows a max-stable distribution. 
It is a reformulation of Example~8.4 in \citet{heffernan2004aconditional} in terms of the stable tail dependence function~$\ell$, allowing for the limit distribution $\nu_1$ to have margins with positive mass at the origin; take for instance $d = 2$ and $\ell(x_1,x_2) = x_1+x_2$.
Since the result is not related to graphical models, we cast it in terms of a random vector $(X_1,\ldots,X_d)$.

\begin{prop}
	\label{prop:nustdf}
Let $X = (X_1,\ldots,X_d)$ have a max-stable distribution $G$ with unit-Fréchet margins and stable tail dependence function $\ell$. If $\ell$ has a continuous first-order partial derivative $\dot{\ell}_1$ with respect to its first argument, then
\[
	\left( \frac{X_j}{t}, \, j \in \{2,\ldots,d\} \mathrel{\Big|} X_1 = t \right)
	\dto
	\nu_{1}, \qquad t \to \infty,
\]
where $\nu_1$ is a probability distribution with support contained in $[0, \infty)^{d-1}$ and determined by
\begin{equation} \label{eqn:Xhr}
	\forall x \in (0, \infty)^{d-1}, \qquad
	\nu_{1}([0, x]) = \dot{\ell}_{1} (1, 1/x_2, \ldots, 1/x_d).
\end{equation}
\end{prop}

The following example will play a key role in the next section. It provides the form of the probability measure $\nu_1$ in the special case the clique vectors $X_C$ follow a max-stable Hüsler--Reiss distribution.

\begin{example}\label{exampleHR}
	Suppose $X=(X_1, \ldots, X_d)$ follows a max-stable Hüsler--Reiss distribution with unit-Fréchet margins and parameter matrix $\Delta=\{\delta_{ij}^2\}$; see Section~\ref{ssec:def_hr_model} below for a more detailed description. The corresponding stable tail dependence function $\ell$ has a continuous first-order partial derivatives, so that $X/t\mid X_1=t$ converges weakly to some distribution $\nu_1$ determined by $\ell$ via~\eqref{eqn:Xhr}. By Remark~\ref{rmk:XXA}, the limit must be the same as the one of $(X/X_1 \mid X_1 > t)$ as $t \to \infty$.
	But Theorem~2 in \citet{engelke2014estimation} states that, as $t \to \infty$,  
	\begin{equation*}
	 \label{eqn:log_hr}
		\bigl( \ln X_j - \ln X_1, \, j =2, \ldots, d \mid X_1 > t \bigr)
		\dto
		\mathcal{N}_{d-1} \bigl(
		\mu_1(\Delta), \Psi_{1}(\Delta)
		\bigr),
	\end{equation*} 
for a mean vector $\mu_1=-2(\delta^2_{1i}, i=2, \ldots, d)$ and covariance matrix 
\[
(\Psi_1)_{ij}=2(\delta_{1i}^2+\delta_{1j}^{2}-\delta_{ij}^2), \qquad i,j=2, \ldots, d.
\]
Hence $X/X_1\mid X_1>t$ and $X/X_1\mid X_1=t$ converge to a multivariate log-normal distribution with the same parameters, $\mu_1(\Delta)$ and $\Psi_1(\Delta)$.
The limit is confirmed in \citet[Example~4]{segers2020one} when $X$ has only two elements.
 
Corollary~\ref{cor:max-stable:stdf} requires that $\E[Z_e]=1$. In the current example, $(X/X_1)\mid X_1>t$ and $(X/t)\mid X_1=t$ converge weakly to the log-normal random vector $(Z_{1i}, i = 2,\ldots,d)$. Then we have $\E[Z_{1i}]=\exp(\mu_{1i}+(\Psi_1)_{ii}/2)$ with $\mu_{1i}=-2\delta^2_{1i}$ and $(\Psi_1)_{ii}=2(\delta^2_{1i}+\delta^2_{1i}+0)=4\delta^2_{1i}$, confirming that $\E[Z_{1i}] = 1$ for every $i=2, \ldots, d$. 
 \hfill$\diamondsuit$
\end{example}


\section{Cliques in the Hüsler--Reiss domain of attraction}
\label{sec:hr}

We will apply Theorem~\ref{prop_main} to a Markov random field $X$ with respect to a block graph $\G = (V, E)$ such that for every (maximal) clique $C \in \mathcal{C}$, the sub-vector $X_C = (X_v, v \in C)$ satisfies Assumption~\ref{ass:nu} with $\nu_{C, s}$ being the one as in Example~\ref{exampleHR}, i.e., a multivariate log-normal distribution with mean $\mu_{C,s}(\Delta_C)$ and covariance matrix $\Psi_{C,s}(\Delta_C)$.

In Proposition~\ref{prop:log-excess}, we find that the limit random vector $A_u$ in Theorem~\ref{prop_main} is multivariate log-normal with mean vector and covariance matrix related to the graph structure.
Moreover, $X$ is in the max-domain of attraction of a max-stable Hüsler--Reiss whose parameter matrix can be derived in a simple way from the matrices $\Delta_C$ (Proposition~\ref{prop_HRattractor}).
Further, the associated multivariate Pareto distribution is an extremal graphical model in the sense of \citet{engelke2020graphical} and this with respect to the same block graph $\G$ (Proposition~\ref{prop:extr_gm}).
The elegant form of the parameter matrix makes this family a suitable candidate for modelling extremes of asymptotically dependent distributions; see also the discussion in \citet{strokorb2020extremal} to \citet{engelke2020graphical} on the issue of extremal independence and disconnected graphs.
Finally, we show that the parameters of the limiting max-stable Hüsler--Reiss distribution are still identifiable in case some variables are latent, and this if and only if every node with a latent variable belongs to at least three different cliques (Proposition~\ref{prop:identif_hr}).
The proofs of the results in this section are given in Appendix~\ref{app:hr}.

\subsection{Max-stable Hüsler--Reiss distribution} \label{ssec:def_hr_model}

The max-stable Hüsler--Reiss distribution arises as the limiting distribution of normalized component-wise sample maxima of a triangular array of row-wise independent and identically distributed Gaussian random vectors with correlation matrix that depends on the sample size \citep{husler1989maxima}.
The Gaussian distribution is in the max-domain of attraction of the Gumbel distribution, but here we transform the margins to the unit-Fréchet distribution.
Let $\pnorm$ denote the standard normal cumulative distribution function.
Recall from~\eqref{eqn:stdf} the stdf $\ell$ of a general max-stable distribution $G$.

The stdf of the bivariate Hüsler--Reiss distribution with parameter $\delta \in (0, \infty)$ is
\begin{equation}
\label{eq:stdf:HR:2}
	\ell_\delta(x, y) = x \pnorm\left(\delta+\frac{\ln(x/y)}{2\delta}\right)
	+ y \pnorm\left(\delta+\frac{\ln(y/x)}{2\delta}\right),
	\qquad (x, y) \in (0, \infty)^2,
\end{equation}
with obvious limits as $x \to 0$ or $y \to 0$.
The boundary cases $\delta \to \infty$ and $\delta \to 0$ correspond to independence, $\ell_\infty(x,y) = x+y$, and co-monotonicity, $\ell_0(x, y) = \max(x, y)$, respectively.
The limit distribution in Proposition~\ref{prop:nustdf} can be calculated explicitly and is equal to the one of the log-normal random variable $\exp\{2 \delta (Z - \delta)\}$, with $Z$ a standard normal random variable \citep[Example~4]{segers2020one}.

\bgroup
\color{darkgray}
\egroup

To introduce the multivariate Hüsler--Reiss distribution, we follow the exposition in \citet{engelke2014estimation}. 
Let $W$ be a finite set with at least two elements and let $\rho(1), \rho(2), \ldots$ be a sequence of $W$-variate correlation matrices, i.e., $\rho(n) = (\rho_{ij}(n))_{i,j \in W}$. 
Assume the limit matrix $\Delta = (\delta_{ij}^2)_{i,j \in W}$ -- denoted by $\Lambda$ in the cited article -- exists:
\begin{equation}
\label{eq:Delta}
	\lim_{n \to \infty} \bigl( 1 - \rho_{ij}(n) \bigr) \ln(n) 
	= \delta_{ij}^2, \qquad i, j \in W.
\end{equation}
Obviously, the matrix $\Delta \in [0, \infty)^{W \times W}$ is symmetric and has zero diagonal.
Suppose further that $\Delta$ is conditionally negative definite, i.e., we have $a^\top \Delta a < 0$ for every non-zero vector $a \in \reals^W$ such that $\sum_{j \in W} a_j = 0$.
[Note that the weak inequality $a^\top \Delta a \le 0$ automatically holds for such $a$ and for the limit $\Delta$ in~\eqref{eq:Delta}.]
For $J\subseteq W$ with $|J| \ge 2$ and for $s\in J$ let $\Psi_{J,s}$ 
be the positive definite, $|J\setminus s|$-square symmetric matrix with elements
\begin{equation} \label{eq:Psi_mat}
	\bigl(\Psi_{J,s}(\Delta)\bigr)_{i,j}
	= 2(\delta_{si}^2+\delta_{sj}^2-\delta_{ij}^2),
	\qquad i, j \in J \setminus s.
\end{equation}
The $|W|$-variate Hüsler--Reiss max-stable distribution with unit-Fréchet margins and parameter matrix $\Delta$ is 
\begin{equation}
\label{eq:HRDelta}
	H_\Delta(x) 
	= \exp \left\{ \sum_{j=1}^{|W|}(-1)^j \sum_{J\subseteq W : |J|=j} h_{\Delta, J}(x_J) \right\}, \qquad x \in (0, \infty)^W,
\end{equation}
with $h_{\Delta, J}(x_J) = 1/x_w$ if $J = \{w\}$, while, if $|J| \ge 2$,
\[
	h_{\Delta, J}(x_J) =
	\int_{\ln(x_s)}^\infty 
	\P\left[ 
		\forall w \in J \setminus s, \, Y_{sw} > \ln(x_w)-z+2\delta_{sw}^2 
	\right]
	e^{-z} \, \diff z
\]
where $s$ can be any element of $J$ and where $Y_s = (Y_{sw}, w \in J \setminus s)$ is a multivariate normal random vector with zero mean vector and covariance matrix $\Psi_{J,s}(\Delta)$ in~\eqref{eq:Psi_mat}.

A shorter expression for $H_\Delta$ is given in \citet[Remark~2.5]{nikoloulopoulos2009extreme}, later confirmed as the finite-dimensional distributions of max-stable Gaussian and Brown-Resnick processes in \citet{genton2011on} and \citet{huser2013composite} respectively:
\[
	H_\Delta(x) = \exp \left\{
		- \sum_{s \in W} \frac{1}{x_s} \pnorm_{|W|-1} \left(
			2 \delta_{vs}^2 + \ln (x_v/x_s), v \in W \setminus s;
			\Psi_{W,s}(\Delta)
		\right)
	\right\},
	\qquad x \in (0, \infty)^W,
\]
with $\pnorm_{d}(\,\cdot\,;\Sigma)$ the $d$-variate normal cdf with covariance matrix $\Sigma$. 
The stdf is thus
\[
	\ell_\Delta(y) = \sum_{s \in W} y_s \pnorm_{|W|-1} \left(
		2 \delta_{vs}^2 + \ln (y_s/y_v), v \in W \setminus s;	\Psi_{W,s}(\Delta) 
	\right),
	\qquad y \in (0, \infty)^W.
\]
If $|W| = 2$ and if the off-diagonal element of $\Delta$ is $\delta^2 \in (0, \infty)$, say, we have $\Psi_{W,s}(\Delta) = 4 \delta^2 = (2 \delta)^2$ and the  stdf $\ell_\Delta$ indeed simplifies to $\ell_\delta$ in~\eqref{eq:stdf:HR:2}.

\bgroup
\color{darkgray}
\egroup


\bgroup
\egroup

\subsection{Hüsler--Reiss limits and extremal graphical models}

Recall from Example~\ref{exampleHR} that for a max-stable Hüsler--Reiss vector $X = (X_1,\ldots,X_d)$, the limit of $X/t$ given $X_1 = t$ as $t \to \infty$ is multivariate log-normal. Now we take that limit as starting point for the tails of the clique vectors of a Markov random field on a block graph.
	
\begin{assum}[Markov block graph with clique-wise Hüsler--Reiss limits]
		\label{ass:HRMf}
		Let $X$ be a Markov random field with respect to the (connected) block graph $\G = (V, E)$ with (maximal) cliques $\mathcal{C}$. Suppose the margins of $X$ satisfy $t \P(X_v > t) \to 1$ as $t \to \infty$ for all $v \in V$.
		For every clique $C \in \mathcal{C}$, let $\Delta_C = (\delta_{ij}^2)_{i,j \in C}$ be the parameter matrix of a $|C|$-variate max-stable Hüsler--Reiss distribution, i.e., $\Delta_C \in [0, \infty)^{C \times C}$ is symmetric, conditionally negative definite, and has zero diagonal.
	 For every $C\in \mathcal{C}$ let $X_C$ satisfy Assumption~\ref{ass:nu} where $\nu_{C,u}$ is the limit in Example~\ref{exampleHR}, i.e., a $|C\setminus u|$-variate log-normal distribution with mean vector 
	 \[
	 \mu_{C,u}=-2(\delta_{ui}^2, i\in C\setminus u)
	 \] 
	 and covariance matrix
	 \[
	 (\Psi_{C,u})_{ij}=2(\delta_{ui}^2+\delta_{uj}^2-\delta_{ij}^2),
	  \qquad i,j\in C\setminus u.
	 \]
	\end{assum}

Because for any $u\in C$ the random vector $X_C$ satisfies the limit in Assumption~\ref{ass:nu} by Corollary~1 and Theorem~2 in \citet{segers2020one} it follows that it is in the max-stable domain of attraction of a Hüsler--Reiss distribution with parameter matrix $\Delta_C$. 

We apply Theorem~\ref{prop_main} to study the limit of $X/t\mid X_u=t$ for some $u\in V$, i.e., the conditional distribution of the field given that it is large at a particular node.
Write $\Delta = (\Delta_C, C \in \mathcal{C})$ and consider the matrix $P(\Delta) = \bigl(p_{ij}(\Delta)\bigr)_{i,j \in V}$ of path sums
\begin{equation}
\label{eq:pij}
	p_{ij}(\Delta) := \sum_{e\in \pth{i}{j}}\delta^2_{e},
\end{equation}
where $\pth{i}{j}$ is the collection of edges on the unique shortest path from $i$ to $j$ and where $\delta^2_{e}$ is to be read off from the matrix $\Delta_C$ for the unique clique $C$ containing the two nodes connected by $e$; by convention, $p_{ii}(\Delta) = 0$ for all $i \in V$, being the sum over the empty set $\pth{i}{i}=\varnothing$.
Let $\mathcal{N}_r(\mu, \Sigma)$ denote the $r$-variate normal distribution with mean vector $\mu$ and covariance matrix $\Sigma$.

\begin{prop}[Logarithm of the limiting field] \label{prop:log-excess}
	Under Assumption~\ref{ass:HRMf}, we have, for each $u \in V$ and as $t \to \infty$, 
	\begin{equation*} 
	\bigl( \ln (X_v/t), \, v\in V \setminus u \mid X_u = t \bigr)
	\dto 
	\mathcal{N}_{|V\setminus u|} \bigl( 
		\mu_{u}(\Delta), \Sigma_{u}(\Delta)
	\bigr)
	\end{equation*}
	with mean vector and covariance matrix written in terms of $p_{ij} = p_{ij}(\Delta)$ in~\eqref{eq:pij} by
\begin{align} 
\label{eqn:muVu}
	\bigl(\mu_{u}(\Delta)\bigr)_i
	&=
	-2p_{ui}, &i \in V \setminus u,\\
	\label{eqn:sigmaVu}
	\bigl(\Sigma_{u}(\Delta)\bigr)_{i,j}
&=
	2(p_{ui}+p_{uj}-p_{ij}), &i, j \in V \setminus u,
\end{align}
and in particular $(\Sigma_u(\Delta))_{i,i} = 4p_{ui}$ for $i \in V \setminus u$.
The matrix $\Sigma_u(\Delta)$ is positive definite and the matrix $P(\Delta)$ is conditionally negative definite.
\end{prop}


\begin{example}
\label{ex:simgraph}
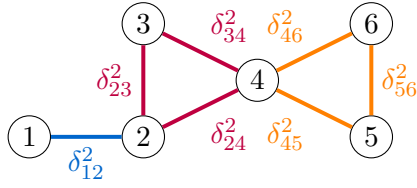
\begin{figure}
	\centering
	\begin{tikzpicture} 
	\node[hollow node] (0) at (0,0) {$1$};
	\node[hollow node] (1) at (1.5,0)  {$2$};
	\node[hollow node] (2) at (1.5,1.5)  {$3$};
	\node[hollow node] (3) at (3, 0.75)  {$4$};
	\node[hollow node] (6) at (4.5,0)  {$5$};
	\node[hollow node] (4) at (4.5,1.5)  {$6$};
	\path (0) edge[color=bbl, line width=1.5] (1)  ;
	\path (1) edge[color=purple, line width=1.5] (2);
	\path (2) edge[color=purple, line width=1.5] (3)  ;
	\path (3) edge[color=orange, line width=1.5] (4);
	\path (3) edge[color=purple, line width=1.5] (1);
	\path (3) edge[color=orange, line width=1.5] (6);
	\path (4) edge[color=orange, line width=1.5] (6);
	\path (0) -- (1) node [midway,auto=right, color=bbl] {$\delta_{12}^2$};
	\path (1) -- (2) node [midway,auto=left, color=purple] {$\delta_{23}^2$};
	\path (3) -- (2) node [midway,auto=right, color=purple] {$\delta_{34}^2$};
	\path (1) -- (3) node [midway,auto=right, color=purple] {$\delta_{24}^2$};
	\path (3) -- (6) node [midway,auto=right, color=orange] {$\delta_{45}^2$};
	\path (3) -- (4) node [midway,auto=left, color=orange] {$\delta_{46}^2$};
	\path (4) -- (6) node [midway,auto=left, color=orange] {$\delta_{56}^2$};
	\end{tikzpicture}
	\caption{The random vector $X$ is six-variate and it is Markov with respect to the graph which contains three cliques, $C_1=\{1,2\}, C_2=\{2,3,4\}$ and $C_3=\{4,5,6\}$. According to Example~\ref{exampleHR} for $C_1$ we have a limiting probability measure $\nu_{C_1, \cdot}$ which depends only on the parameter $\textcolor{bbl}{\delta_{12}}$; for $C_2$ we have  $\nu_{C_2, \cdot}$ which depends only on the parameters $\textcolor{purple}{\delta_{23}^2, \delta_{24}^2, \delta_{34}^2}$ and for $C_3$ we have  $\nu_{C_3, \cdot}$ which depends only on the parameters $\textcolor{orange}{\delta_{45}^2, \delta_{46}^2, \delta_{56}^2}$. }
	\label{fig:simgraph}
\end{figure}
Consider a Markov field with respect to the block graph in Figure~\ref{fig:simgraph}. 
The graph has three cliques, to which correspond three Hüsler--Reiss limits with  parameter matrices respectively
\begin{equation*}
\textcolor{bbl}{\Delta_1}=\begin{bmatrix}0& \textcolor{bbl}{\delta_{12}^2}\\
&0\end{bmatrix},\quad
\textcolor{purple}{\Delta_2}=\begin{bmatrix}0&\textcolor{purple}{\delta_{23}^2}
&\textcolor{purple}{\delta_{24}^2}\\
&0&\textcolor{purple}{\delta_{34}^2}\\
&
& 0\end{bmatrix}, \quad
\textcolor{orange}{  	\Delta_3}=\begin{bmatrix}0&\textcolor{orange}{\delta_{45}^2}&
\textcolor{orange}{\delta_{46}^2}\\
&0&\textcolor{orange}{\delta_{56}^2}\\
&
& 0\end{bmatrix}.
\end{equation*}
If a high threshold is exceeded at node~$u = 1$, the limiting $5$-variate normal distribution in Proposition~\ref{prop:log-excess} has means $(\mu_1(\Delta))_i$ and variances $(\Sigma_1(\Delta))_{ii}$ proportional to the path sums $p_{1i} = \sum_{e \in \pth{1}{i}} \delta_{e}^2$ for $i \in \{2,\ldots,5\}$, while the off-diagonal entries of the covariance matrix are given by
\begin{align*}
	\bigl( \Sigma_1(\Delta) \bigr)_{2,j} &= 4 \textcolor{bbl}{\delta_{12}^2}, && j \in \{3, \ldots, 6\}, \\
	\bigl( \Sigma_1(\Delta) \bigr)_{3,j} &= 4 \left(\textcolor{bbl}{\delta_{12}^2} + \tfrac{1}{2} (\textcolor{purple}{\delta_{23}^2 + \delta_{24}^2 - \delta_{34}^2})\right), 	
	&& j \in \{4, 5, 6\}, \\
	\bigl( \Sigma_1(\Delta) \bigr)_{4,j} &= 4 (\textcolor{bbl}{\delta_{12}^2} + \textcolor{purple}{\delta_{24}^2}), && j \in \{5, 6\}, \\	
	\bigl( \Sigma_1(\Delta) \bigr)_{5,6} &= 4 \left(\textcolor{bbl}{\delta_{12}^2} + \textcolor{purple}{\delta_{24}^2} + \tfrac{1}{2}(\textcolor{orange}{\delta_{45}^2 + \delta_{46}^2 - \delta_{56}^2})\right).	
\end{align*}
The dependence within blocks is visible in the covariances at entries $(3, j)$ for $j \in \{4, 5, 6\}$ and the one at entry $(5, 6)$. \hfill$\diamondsuit$
\end{example}



The exact distribution of the Markov field with clique-wise Hüsler--Reiss limits in Assumption~\ref{ass:HRMf} is 
in the max-domain of attraction of such a distribution.

\begin{prop}[Max-domain of attraction
	] \label{prop_HRattractor}
	The Markov random field $X$ in Assumption~\ref{ass:HRMf} is in the max-domain of attraction of the Hüsler--Reiss max-stable distribution~\eqref{eq:HRDelta} with unit-Fréchet margins and parameter matrix $P(\Delta)$ in~\eqref{eq:pij}, that is,
	\begin{equation*}
		\lim_{n \to \infty}
		\bigl( \P\left( \forall v \in V : X_v \le n x_v \right) \bigr)^n
		= H_{P(\Delta)}(x),
		\qquad x \in (0, \infty)^V.
	\end{equation*}
\end{prop}

Recall from Section~\ref{sec:evt-intro} that because $X$ in Proposition~\ref{prop_HRattractor} belongs to the domain of attraction of the max-stable distribution $H_{P(\Delta)}$, the asymptotic distribution of the vector of high-threshold excesses is a multivariate Pareto distribution determined by $H_{P(\Delta)}$ via \eqref{eqn:pareto_lim}. The latter is the distribution of the random vector $Y$ in the next proposition and is called a Hüsler--Reiss Pareto distribution 
in \citet{engelke2020graphical}.
The distribution of $Y$ turns out to be an extremal graphical model in the sense of \citet[Definitions~1 and~2]{engelke2020graphical}. We recall this notion here. Let $Y$ be a multivariate Pareto random vector in~\eqref{eqn:pareto} and for $u \in V$, let $Y^{(u)}$ be a random vector equal in distribution to $Y \mid Y_u > 1$.
Then $Y$ is an extremal graphical model with respect to a graph $\G=(V,E)$ if we have conditional independence $Y^{(u)}_i\indep Y^{(u)}_j \mid Y^{(u)}_{V\setminus \{u,i,j\}}$ for all $i,j\in V\setminus u$ such that $(i,j) \notin E$.  

\begin{prop}[Attraction to extremal graphical model] \label{prop:extr_gm} 
	The Markov block graph $X$ in Assumption~\ref{ass:HRMf} satisfies the weak convergence relation~\eqref{eqn:pareto} with $Y$ distributed as in~\eqref{eqn:pareto_lim} for $G = H_{P(\Delta)}$, the limit in Proposition~\ref{prop_HRattractor}.
	This $Y$ is an extremal graphical model with respect to $\G$ in the sense of \citet[Definition~2]{engelke2020graphical}.
\end{prop} 

Proposition~\ref{prop:extr_gm} leads to the elegant result that the graphical model $X$ obtained by endowing every clique $C$ of a block graph $\G$ by a limit based on the Hüsler--Reiss max-stable distribution with parameter matrix $\Delta_C$ is in the Pareto domain of attraction of a Hüsler--Reiss Pareto random vector $Y$ which is itself an extremal graphical model with respect to the same graph $\G$ and with, on every clique $C$, a Hüsler--Reiss Pareto distribution with the same parameter matrix $\Delta_C$.
In other words, the Pareto limit of a graphical model constructed clique-wise by distributions with Hüsler--Reiss limits is an extremal graphical model constructed clique-wise by Hüsler--Reiss Pareto distributions.

Proposition~\ref{prop:extr_gm} also sheds new light on Proposition~4 in \citet{engelke2020graphical}, where the existence and uniqueness of a Hüsler--Reiss extremal graphical model was established given the Hüsler--Reiss distributions on the cliques of a block graph. In our construction, the solution is explicit and turns out to have the simple and elegant form in terms of the path sums $p_{ij}(\Delta)$ in~\eqref{eq:pij}.

\subsection{Latent variables and parameter identifiability}

In \citet*{asenova2021inference} a criterion was presented for checking whether the parameters of the Hüsler--Reiss distribution are identifiable if for some of the nodes $v \in V$ the variables $X_v$ are unobservable (latent).
The issue was illustrated for river networks when the water level or another variable of interest is not observed at some splits or junctions. 
For trees, a necessary and sufficient identifiability criterion was that every node with a latent variable should have degree at least three.

For block graphs, a similar condition turns out to hold. The degree of a node $v \in V$, i.e., the number of neighbours, is now replaced by its \emph{clique degree}, notation $\cd(v)$, defined as the number of cliques containing that node.

Let the setting be the same as in Proposition~\ref{prop_HRattractor} and let $H_{P(\Delta)}(x)$ be the $|V|$-variate max-stable Hüsler--Reiss distribution with parameter matrix $P(\Delta)$ in~\eqref{eq:pij}.
Let the (non-empty) set of nodes with observable variables be $U \subset V$, so that $\bar{U} = V \setminus U$ is the set of nodes with latent variables.
As the max-stable Hüsler--Reiss family is stable under taking marginals \citep[Example~7]{engelke2020graphical}, the vector $X_U = (X_v, v \in U)$ is in the max-domain of attraction of the $|U|$-variate max-stable Hüsler--Reiss distribution with parameter matrix $P(\Delta)_U = (p_{ij}(\Delta))_{i,j \in U}$.
If $\bar{U}$ is non-empty, $U$ is a proper subset of $V$, and the question is whether we can reconstruct the whole matrix $P(\Delta)$ given only the sub-matrix $P(\Delta)_U$ and the graph $\G$.
Note that the entries in $P(\Delta)_U$ are the path sums between nodes carrying observable variables only.
The question is whether we can find the other path sums too, that is, those between nodes one or two of which carry latent variables.


\begin{prop}[Identifiability] \label{prop:identif_hr}
	Given the block graph $\G = (V, E)$ and node sets $\bar{U} \subset V$ and $U=V\setminus \bar{U}$, the Hüsler--Reiss parameter matrix $P(\Delta)$ in~\eqref{eq:pij} is identifiable from the restricted matrix $P(\Delta)_U = (p_{ij}(\Delta))_{i,j \in U}$ if and only if $\cd(v) \ge 3$ for every $v \in \bar{U}$.
\end{prop}

\begin{example}
\label{ex:identif_hr}
\begin{figure} 
	\centering 
	\begin{tikzpicture} \label{fig:sim_hr_mis}
	\node[hollow node] (a) at (-0.75,1.5)  {$1$};
	\node[hollow node] (b) at (-1.5,0)  {$2$};
	\node[hollow node, color=red, line width=1] (c) at (0,0)  {$3$};
	\node[hollow node] (d) at (0.75,1.5)  {$4$};
	\node[hollow node] (e) at (1.5,0)  {$5$};
	\node[hollow node] (f) at (-0.75, -1.5) {$6$};
	\node[hollow node] (g) at (0.75, -1.5) {$7$};
	\path (a) edge[color=bbl, line width=1.5] (b);
	\path (c) edge[color=bbl, line width=1.5] (a)  ;
	\path (c) edge[color=bbl, line width=1.5] (b);
	\path (c) edge[color=ggr, line width=1.5] (d);
	\path (c) edge[color=ggr, line width=1.5] (e);
	\path (d) edge[color=ggr, line width=1.5] (e);
	\path (c) edge[color=ppn, line width=1.5] (f);
	\path (c) edge[color=ppn, line width=1.5] (g);
	\path (f) edge[color=ppn, line width=1.5] (g);
	\end{tikzpicture}
	\caption{A block graph with three cliques. In the first clique with node set $C_1=\{1,2,3\}$ the parameters are $\textcolor{bbl}{\delta_{12}^2, \delta_{13}^2, \delta_{23}^2}$, in the second clique $C_1=\{3,4,5\}$ the parameters are $\textcolor{ggr}{\delta_{34}^2, \delta_{35}^2, \delta_{45}^2}$, and in the third clique $C_1=\{3,6,7\}$ the parameters are $\textcolor{ppn}{\delta_{36}^2, \delta_{37}^2, \delta_{67}^2}$. These nine parameters determine the Hüsler--Reiss parameter matrix $P(\Delta)$ in~\eqref{eq:pij}. The nine parameters and thus the entire matrix $P(\Delta)$ is identifiable from the submatrix $P(\Delta)_U$ with $U = V \setminus \bar{U}$ for $\bar{U}= \{3\}$ because node $3$ belongs to three different cliques (Proposition~\ref{prop:identif_hr} and Example~\ref{ex:identif_hr}).}
	\label{fig:hr_sim}
\end{figure}
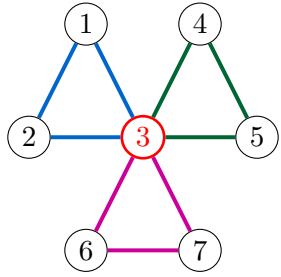

Consider the block graph $\G = (V, E)$ in Figure~\ref{fig:hr_sim} with $V = \{1,\ldots,7\}$ and $\bar{U}=\{3\}$. Therefore $U = V \setminus \{3\}$.
Node $v = 3$ belongs to three different cliques and thus has clique degree $\cd(v) = 3$.
By Proposition~\ref{prop:identif_hr}, all edge parameters $\delta_e^2$ for $e \in E$ can be identified from the path sums $p_{ij}$ for $i, j \in U$.
Indeed, for the edges $e = (a, b) \in \{(1, 2), (4, 5), (6, 7)\}$ this follows from the identity $\delta_{e}^2 = p_{ab}$, while for the edges $\delta_{i,3}^2$ for $i \ne 3$ this follows from a calculation such as
\[
	\delta_{13}^2 = \frac{1}{2} \left( p_{14} + p_{16} - p_{46} \right).
\]
If, however, node $v=1$ would not belong to $U$, then the edge parameters $\delta_{12}^2$ and $\delta_{13}^2$ would not be identifiable from the path sums $p_{ij}$ for $i, j \in V \setminus \{1\}$, since none of these paths contains edges $(1, 2)$ or $(1, 3)$.
\end{example}

\section{Conclusion}
\label{sec:concl}

We have studied the tails of suitably normalized random vectors which satisfy the global Markov property with respect to a block graph. Block graphs are generalizations of trees and this explains why the results presented here are closely related to the ones in \citet{segers2020one} and \citet*{asenova2021inference}. 
The common feature is the existence of a unique shortest path between each pair of nodes. This property is key for our results, although it is not sufficient in itself to explain the multiplicative random walk structure of the limiting field. The latter property is also due to the singleton nature of the minimal clique separators, a property which is no longer present for more general decomposable graphs.
Still, an essential difference between tails of Markov fields with respect to trees on the one hand and more general block graphs on the other hand is that in the latter case, the increments of the random walk in the tail field are dependent within cliques. The regularity assumptions needed for the limit to hold imply multivariate regular variation and thus the model in question is suitable for asymptotically dependent variables.   


We have then focused on a particular random field with respect to a block graph, namely one for which the distribution on each clique satisfies a tail condition based on a Hüsler--Reiss distribution.
We have shown that the logarithm of the limiting field is a normal random vector with mean and covariance matrix that depend on the sums of the edge weights along the unique shortest paths between pairs of nodes.
The same structural pattern emerges in the parameter matrix of the max-stable Hüsler--Reiss distribution to which the Markov field is attracted.
The relation between the original Markov field as an ordinary graphical model on the one hand and Hüsler--Reiss extremal graphical models as in \citet{engelke2020graphical} on the other hand was highlighted.
Due to the path sum structure of the parameter matrix, all edge weights remain identifiable even when variables associated to nodes with clique degree at least three are latent.


An interesting problem would be to identify a minimal requirement on a graph that leads to the multiplicative structure of the tail field of a Markov field that we found for block graphs. 
Another question is which structure replaces the multiplicative random walk form for more general graphs, for instance decomposable graphs, and what this means for specific parametric families. Another research direction could be the study of the tails of the Markov field under assumptions related to hidden regular variation and/or allowing for transitions from non-extreme to extreme regions.   

\section*{Acknowledgements}
Stefka Asenova wishes to thank Antonio Abruzzo for his much appreciated help and suggestions on graphical models.
Johan Segers wishes to thank Sebastian Engelke and Manuel Hentschel for inspiring discussions on Hüsler--Reiss distributions. We are highly grateful to two anonymous Reviewers and the Associate Editor for their suggestions and for pointing us to literature that has enriched our perception on related topics.

\section*{Data availability statement}
There is no data used in this manuscript.

\section*{Conflict of interest}
The authors declare that they have no conflict of interest.

	
\bibliography{bibliography_pr3}

\begin{thebibliography}{41}
\providecommand{\natexlab}[1]{#1}
\providecommand{\url}[1]{\texttt{#1}}
\expandafter\ifx\csname urlstyle\endcsname\relax
  \providecommand{\doi}[1]{doi: #1}\else
  \providecommand{\doi}{doi: \begingroup \urlstyle{rm}\Url}\fi

\bibitem[Am{\'e}ndola et~al.(2022)Am{\'e}ndola, Kl{\"u}ppelberg, Lauritzen, and
  Tran]{amendola2022conditional}
C.~Am{\'e}ndola, C.~Kl{\"u}ppelberg, S.~Lauritzen, and N.~M. Tran.
\newblock Conditional independence in max-linear {Bayesian} networks.
\newblock \emph{The Annals of Applied Probability}, 32\penalty0 (1):\penalty0
  1--45, 2022.

\bibitem[Asenova et~al.(2021)Asenova, Mazo, and Segers]{asenova2021inference}
S.~Asenova, G.~Mazo, and J.~Segers.
\newblock Inference on extremal dependence in the domain of attraction of a
  structured {H}{\"u}sler--{R}eiss distribution motivated by a {M}arkov tree
  with latent variables.
\newblock \emph{Extremes}, 24:\penalty0 461--500, 2021.

\bibitem[Behtoei et~al.(2010)Behtoei, Jannesari, and Taeri]{behtoei2010a}
A.~Behtoei, M.~Jannesari, and B.~Taeri.
\newblock A characterization of block graphs.
\newblock \emph{Discrete Applied Mathematics}, 158\penalty0 (3):\penalty0
  219--221, 2010.

\bibitem[Beirlant et~al.(2004)Beirlant, Goegebeur, Segers, and
  Teugels]{beirlant2004statistics}
J.~Beirlant, Y.~Goegebeur, J.~Segers, and J.~L. Teugels.
\newblock \emph{Statistics of {E}xtremes: {T}heory and {A}pplications}, volume
  558.
\newblock John Wiley \& Sons, New Jersey, 2004.

\bibitem[de~Haan and Ferreira(2007)]{haan2007extreme}
L.~de~Haan and A.~Ferreira.
\newblock \emph{Extreme {V}alue {T}heory: {A}n {I}ntroduction}.
\newblock Springer Series in Operations Research and Financial Engineering.
  Springer, New York, 2007.

\bibitem[de~Haan and Resnick(1977)]{haan1977limit}
L.~de~Haan and S.~I. Resnick.
\newblock Limit theory for multivariate sample extremes.
\newblock \emph{Zeitschrift f{\"u}r Wahrscheinlichkeitstheorie und verwandte
  Gebiete}, 40\penalty0 (4):\penalty0 317--337, 1977.

\bibitem[Engelke and Hitz(2020)]{engelke2020graphical}
S.~Engelke and A.~S. Hitz.
\newblock {G}raphical {m}odels for {e}xtremes.
\newblock \emph{Journal of the Royal Statistical Society: Series B (Statistical
  Methodology)}, 82\penalty0 (3):\penalty0 1--38, 2020.

\bibitem[Engelke and Volgushev(2020)]{engelke2020structure}
S.~Engelke and S.~Volgushev.
\newblock Structure learning for extremal tree models, 2020.
\newblock \url{https://arxiv.org/abs/2012.06179v2}.

\bibitem[Engelke et~al.(2014)Engelke, Malinowski, Kabluchko, and
  Schlather]{engelke2014estimation}
S.~Engelke, A.~Malinowski, Z.~Kabluchko, and M.~Schlather.
\newblock Estimation of {H}üsler--{R}eiss distributions and {B}rown--{R}esnick
  processes.
\newblock \emph{Journal of the Royal Statistical Society: Series B (Statistical
  Methodology)}, 77\penalty0 (1):\penalty0 239--265, 2014.

\bibitem[Engelke et~al.(2022)Engelke, Ivanovs, and
  Strokorb]{strokorb2022graphical}
S.~Engelke, J.~Ivanovs, and K.~Strokorb.
\newblock Graphical models for infinite measures with applications to extremes
  and {Lévy} processes, 2022.
\newblock URL \url{https://arxiv.org/abs/2211.15769}.

\bibitem[Genton et~al.(2011)Genton, Ma, and Sang]{genton2011on}
M.~G. Genton, Y.~Ma, and H.~Sang.
\newblock {On the likelihood function of Gaussian max-stable processes}.
\newblock \emph{Biometrika}, 98\penalty0 (2):\penalty0 481--488, 2011.

\bibitem[Gissibl and Klüppelberg(2018)]{gissibl2018max}
N.~Gissibl and C.~Klüppelberg.
\newblock Max-linear models on directed acyclic graphs.
\newblock \emph{Bernoulli}, 24\penalty0 (4A):\penalty0 2693--2720, 2018.

\bibitem[Heffernan and Resnick(2005)]{resnick2005hidden}
J.~Heffernan and S.~Resnick.
\newblock Hidden regular variation and the rank transform.
\newblock \emph{Advances in Applied Probability}, 37\penalty0 (2):\penalty0
  393--414, 2005.

\bibitem[Heffernan and Resnick(2007)]{resnick2007limit}
J.~E. Heffernan and S.~I. Resnick.
\newblock {Limit laws for random vectors with an extreme component}.
\newblock \emph{The Annals of Applied Probability}, 17\penalty0 (2):\penalty0
  537--571, 2007.

\bibitem[Heffernan and Tawn(2004)]{heffernan2004aconditional}
J.~E. Heffernan and J.~A. Tawn.
\newblock A conditional approach for multivariate extreme values (with
  discussion).
\newblock \emph{Journal of the Royal Statistical Society: Series B (Statistical
  Methodology)}, 66\penalty0 (3):\penalty0 497--546, 2004.

\bibitem[Hernandez-Campos et~al.(2005)Hernandez-Campos, Jeffay, Park, Marron,
  and Resnick]{campos2005extremal}
F.~Hernandez-Campos, K.~Jeffay, C.~Park, J.~S. Marron, and S.~I. Resnick.
\newblock Extremal dependence: {I}nternet traffic applications.
\newblock \emph{Stochastic Models}, 21\penalty0 (1):\penalty0 1--35, 2005.

\bibitem[Huser and Davison(2013)]{huser2013composite}
R.~Huser and A.~C. Davison.
\newblock {Composite likelihood estimation for the {B}rown–{R}esnick
  process}.
\newblock \emph{Biometrika}, 100\penalty0 (2):\penalty0 511--518, 2013.

\bibitem[Hüsler and Reiss(1989)]{husler1989maxima}
J.~Hüsler and R.-D. Reiss.
\newblock Maxima of normal random vectors: {B}etween independence and complete
  dependence.
\newblock \emph{Statistics \& Probability Letters}, 7\penalty0 (4):\penalty0
  283--286, February 1989.

\bibitem[Janssen and Segers(2014)]{janssen2014markov}
A.~Janssen and J.~Segers.
\newblock Markov tail chains.
\newblock \emph{Journal of Applied Probability}, 51\penalty0 (4):\penalty0
  1133--1153, 2014.

\bibitem[Lauritzen(1996)]{lauritzen1996graphical}
S.~L. Lauritzen.
\newblock \emph{Graphical {M}odels}.
\newblock Oxford University Press, Oxford, 1996.

\bibitem[Le and Tuy(2010)]{le2010the}
V.~B. Le and N.~N. Tuy.
\newblock The square of a block graph.
\newblock \emph{Discrete Mathematics}, 310\penalty0 (4):\penalty0 734--741,
  2010.

\bibitem[Ledford and Tawn(1996)]{ledford1996statistics}
A.~W. Ledford and J.~A. Tawn.
\newblock {Statistics for near independence in multivariate extreme values}.
\newblock \emph{Biometrika}, 83\penalty0 (1):\penalty0 169--187, 1996.

\bibitem[Ledford and Tawn(1997)]{ledford1997modelling}
A.~W. Ledford and J.~A. Tawn.
\newblock Modelling dependence within joint tail regions.
\newblock \emph{Journal of the Royal Statistical Society: Series B (Statistical
  Methodology)}, 59\penalty0 (2):\penalty0 475--499, 1997.

\bibitem[Lee and Joe(2017)]{lee2017multivariate}
D.~Lee and H.~Joe.
\newblock Multivariate extreme value copulas with factor and tree dependence
  structures.
\newblock \emph{Extremes}, 21:\penalty0 1--30, 06 2017.

\bibitem[Nikoloulopoulos et~al.(2009)Nikoloulopoulos, Joe, and
  Li]{nikoloulopoulos2009extreme}
A.~K. Nikoloulopoulos, H.~Joe, and H.~Li.
\newblock Extreme value properties of multivariate $t$ copulas.
\newblock \emph{Extremes}, 12\penalty0 (2):\penalty0 129--148, 2009.

\bibitem[Papastathopoulos and Strokorb(2016)]{strokorb2016conditional}
I.~Papastathopoulos and K.~Strokorb.
\newblock Conditional independence among max-stable laws.
\newblock \emph{Statistics \& Probability Letters}, 108:\penalty0 9--15, 2016.

\bibitem[Papastathopoulos and Tawn(2020)]{papastat2019hidden}
I.~Papastathopoulos and J.~A. Tawn.
\newblock Hidden tail chains and recurrence equations for dependence parameters
  associated with extremes of higher-order {M}arkov chains, 2020.
\newblock \url{https://arxiv.org/abs/1903.04059}.

\bibitem[Papastathopoulos et~al.(2017)Papastathopoulos, Strokorb, Tawn, and
  Butler]{papastat2017extreme}
I.~Papastathopoulos, K.~Strokorb, J.~A. Tawn, and A.~Butler.
\newblock Extreme events of {M}arkov chains.
\newblock \emph{Advances in Applied Probability}, 49\penalty0 (1):\penalty0
  134--161, 2017.

\bibitem[Perfekt(1994)]{perfekt1994extremal}
R.~Perfekt.
\newblock {Extremal Behaviour of Stationary Markov Chains with Applications}.
\newblock \emph{The Annals of Applied Probability}, 4\penalty0 (2):\penalty0
  529--548, 1994.

\bibitem[Resnick(1987)]{resnick1987extreme}
S.~Resnick.
\newblock \emph{Extreme {V}alues, {R}egular {V}ariation, and {P}oint
  {P}rocesses}.
\newblock Springer-Verlag, New York, 1987.

\bibitem[Resnick(2002)]{resnick2002hidden}
S.~Resnick.
\newblock Hidden regular variation, second order regular variation and
  asymptotic independence.
\newblock \emph{Extremes}, 5:\penalty0 303--336, 2002.

\bibitem[Resnick and Zeber(2013)]{resnick2013asymptotics}
S.~I. Resnick and D.~Zeber.
\newblock Asymptotics of {M}arkov {k}ernels and the {t}ail {c}hain.
\newblock \emph{Advances in Applied Probability}, 45\penalty0 (1):\penalty0
  186--213, 2013.

\bibitem[Rockafellar and Wets(1998)]{rockafellar1998variational}
R.~T. Rockafellar and R.~J.-B. Wets.
\newblock \emph{Variational Analysis}, volume 317.
\newblock Springer-Verlag, Berlin, 1998.

\bibitem[Rootzén and Tajvidi(2006)]{rootzen2006multivariate}
H.~Rootzén and N.~Tajvidi.
\newblock Multivariate generalized {P}areto distributions.
\newblock \emph{Bernoulli}, 12\penalty0 (5):\penalty0 917--930, 10 2006.

\bibitem[Rootzén et~al.(2018)Rootzén, Segers, and
  Wadsworth]{rootzen2018multivariate}
H.~Rootzén, J.~Segers, and J.~L. Wadsworth.
\newblock Multivariate generalized {Pareto} distributions: {Parametrizations},
  representations, and properties.
\newblock \emph{Journal of Multivariate Analysis}, 165:\penalty0 117--131,
  2018.

\bibitem[Segers(2007)]{segers2007multivariate}
J.~Segers.
\newblock Multivariate regular variation of heavy-tailed {M}arkov chains, 2007.
\newblock \url{https://arxiv.org/abs/math/0701411v1}.

\bibitem[Segers(2020)]{segers2020one}
J.~Segers.
\newblock One-versus multi-component regular variation and extremes of {M}arkov
  trees.
\newblock \emph{Advances in Applied Probability}, 52\penalty0 (3):\penalty0
  855--878, 2020.

\bibitem[Smith(1992)]{smith1992the}
R.~L. Smith.
\newblock The extremal index for a {M}arkov chain.
\newblock \emph{Journal of Applied Probability}, 29\penalty0 (1):\penalty0
  37--45, 1992.

\bibitem[Strokorb(2020)]{strokorb2020extremal}
K.~Strokorb.
\newblock Extremal independence old and new, 2020.
\newblock \url{https://doi.org/10.48550/arXiv.2002.07808}.

\bibitem[van~der Vaart(1998)]{vaart1998asymptotic}
A.~W. van~der Vaart.
\newblock \emph{Asymptotic {S}tatistics}.
\newblock Cambridge University Press, Cambridge, 1998.

\bibitem[Yun(1998)]{yun1998the}
S.~Yun.
\newblock The extremal index of a higher-order stationary {M}arkov chain.
\newblock \emph{Annals of Applied Probability}, pages 408--437, 1998.

\end{thebibliography}
\bibliographystyle{apalike}

\appendix

\section{Proofs for Section~\ref{sec:eot}}
\label{app:eot}

\subsection{Proof of Theorem~\ref{prop_main}}

	The proof follows the lines of the one of Theorem~1 in \citet{segers2020one}. To show~\eqref{eqn:main_res} it is sufficient to show that for a real bounded Lipschitz function $f$, for any fixed $u\in V$ it holds that
	\begin{equation} \label{eq:limitXA}
	\lim_{t \to \infty} \E[f(X_{V\setminus u}/t)\mid X_u=t] = \E[f(A_{u,V\setminus u})],
	\end{equation}
	\citep[Lemma~2.2]{vaart1998asymptotic}.
	Without loss of generality, we assume that $0 \le f(x) \le 1$ and $|f(x)-f(y)|\leq L\sum_{j}|x_j-y_j|$ for some constant $L>0$.

	We proceed by induction on the number of cliques, $m$. When there is only one clique ($m = 1$) the convergence happens by Assumption~\ref{ass:nu} with $s = u$: the distribution of $A_{u,V \setminus u}$ is equal to $\nu_{C,s}$ in the assumption, with $C = V$ and $u = s$.
	
	Assume that there are at least two cliques, $m \ge 2$.
	Let the numbering of the cliques be such that the last clique, $C_m$, is connected to the subgraph induced by $\bigcup_{i=1}^{m-1}C_i$ only through one node, which is the minimal separator between $C_m$ and $\bigcup_{i=1}^{m-1}C_i$. Let $s\in \mathcal{S}$ denote this node and introduce the set 
	\[
		C_{1:m-1}=(C_1 \cup \cdots \cup C_{m-1}) \setminus u.
	\]
	Note that $\{s\} = (C_1 \cup \cdots \cup C_{m-1}) \cap C_m$.
	We need to make a distinction between two cases: $s = u$ or $s \ne u$. The case $s = u$ is the easier one, since then $X_{C_m \setminus u}$ and $X_{C_{1:m-1}}$ are conditionally independent given $X_u$ whereas the shortest paths from $u$ to nodes in $C_m \setminus u$ just consist of single edges, avoiding $C_{1:m-1}$ altogether. So we only consider the case $s \ne u$ henceforth. In that case, paths from $u$ to nodes in $C_m \setminus s$ pass through $s$ and possibly other nodes in $C_{1:m-1}$.
	
	The induction hypothesis is that as $t\rightarrow\infty$, we have
	\begin{equation} \label{eqn:ind_assum}
	(X_{C_{1:m-1}}/t)\mid X_u=t\dto A_{u,C_{1:m-1}},
	\end{equation}
	or also that for every continuous and bounded function $h: \reals_+^{C_{1:m-1}} \to \reals$, we have
	\begin{equation} \label{eqn:ind_assum1}
		\lim_{t \to \infty}
		\E\left[h(X_{C_{1:m-1}}/t)\mid X_u=t\right] =
		\E[h(A_{u,C_{1:m-1}})].
	\end{equation} 
	To prove the convergence in~\eqref{eq:limitXA} we start with the following inequality: for $\delta>0$,
	\begin{align}
	\lefteqn{
	\Bigl| \E[f(X_{V\setminus u}/t)\mid X_u=t]  -\E[f(A_{u,V\setminus u})]\Bigr|} 
	\nonumber\\&\leq
	\Bigl|\E\left[f(X_{V\setminus u}/t) \ind(X_{s}/t\geq\delta)\mid X_u=t\right]
	-
	\E\left[f(A_{u,V\setminus u}) \ind(A_{us}\geq\delta)\right]\Bigr|
	\label{eq:larger}\\&\null +
	\Bigl|\E\left[f(X_{V\setminus u}/t) \ind(X_{s}/t<\delta)\mid X_u=t\right]
	-
	\E\left[f(A_{u,V\setminus u}) \ind(A_{us}<\delta)\right]\Bigr|. \label{eq:smaller}
	\end{align}
	We let $\delta > 0$ be a continuity point of $A_{us}$. Later on, we will take $\delta$ arbitrarily close to zero, which we can do, since the number of atoms of $A_{us}$ is at most countable.
	
	\paragraph{\mdseries\itshape Analysis of~\eqref{eq:larger}.}
	We first deal with~\eqref{eq:larger}. The first expectation is equal to
	\[
	\int_{[0,\infty)^{V\setminus u}}
	f(x/t) \,
	\ind(x_s/t\geq\delta) \,
	\P(X_{V\setminus u}\in \mathrm{d}x \mid X_u=t)\, .
	\]    
	Because of the global Markov property, $X_{C_m\setminus s}$ is conditionally independent of the variables in the set $C_{1:m-1}$ given $X_s$. As a consequence, the conditional distribution of $X_{C_m\setminus s}$ given $X_{C_{1:m-1}}$ is the same as the one of $X_{C_m \setminus s}$ given $X_s$. Hence we can write the integral as 
	\[
	\int_{[0,\infty)^{C_{1:m-1}}}
	\E\left[f\left(x/t, X_{C_m\setminus s}/t\right) \mid X_{s}=x_{s}\right]
	\ind(x_s/t \ge \delta)
	\P(X_{C_{1:m-1}}\in \mathrm{d}x\mid X_u=t).
	\]
	After the change of variables $x/t=y$, the integral becomes
	\begin{equation}
	 \label{eq:int_y}
	\int_{[0,\infty)^{C_{1:m-1}}}
	\E\left[f\left(y, X_{C_m\setminus s}/t\right)
	\, \Big| \,
	X_{s}=ty_{s}\right]
	\ind(y_{s}\geq\delta)
	\P\left(X_{C_{1:m-1}}/t \in \mathrm{d}y \mid X_u=t\right).
	\end{equation}

	Define the functions $g_t$ and $g$ on $[0, \infty)^{C_{1:m-1}}$ by
	\begin{align*}
	g_t(y)
	&:=
	\E\left[f\left(y, X_{C_m\setminus s}/t\right)
	\, \Big| \,
	X_{s}=ty_{s}\right]
	\ind(y_{s}\geq\delta), \\
	g(y)
	&:=
	\E\left[f\left(y, y_sZ_{s,C_m\setminus s}\right)\right]
	\ind(y_{s}\geq\delta).
	\end{align*}
	Consider points $y(t)$ and $y$ in $[0, \infty)^{C_{1:m-1}}$ such that $\lim_{t \to \infty} y(t) = y$ and such that $y_s \ne \delta$. We need to show that 
	\begin{equation}
	\label{eq:gt2g}
		\lim_{t \to \infty} g_t(y(t)) = g(y). 
	\end{equation}
	If $y_s < \delta$, this is clear since $y_s(t) < \delta$ for all large $t$ and hence the indicators will be zero. So suppose $y_s > \delta$ and thus also $y_s(t) > \delta$ for all large $t$, meaning that both indicators are (eventually) equal to one. By Assumption~\ref{ass:nu}, we have
	\[
		X_{C_m \setminus s} / t \mid X_s = ty_s(t)
		\dto
		y_s Z_{s,C_m \setminus s}, \qquad t \to \infty.
	\]
	Since $f$ is continuous, also
	\[
		f\left( y(t), X_{C_{m} \setminus s} / t \right) 
		\, \Big| \, 
		X_s = ty_s(t) 
		\dto f\left(y, y_s Z_{s, C_m \setminus s}\right),
		\qquad t \to \infty.
	\]
	As the range of $f$ is contained in $[0, 1]$, the bounded convergence theorem implies that we can take expectations in the previous equation and conclude~\eqref{eq:gt2g}.

	By the induction hypothesis~\eqref{eqn:ind_assum} and Theorem~18.11 in \citet{vaart1998asymptotic}, the continuous convergence in~\eqref{eq:gt2g} implies
	\begin{equation*}
	g_t\left(\frac{X_{C_{1:m-1}}}{t}\right) \, \Big| \, X_u=t
	\dto g(A_{C_{1:m-1}}),
	\qquad t \to \infty;
	\end{equation*}
	note that by the choice of $\delta$, the discontinuity set of $g$ receives zero probability in the limit.
	As $g_t$ and $g$ are bounded (since $f$ is bounded), we can take expectations and find
		\begin{equation} \label{eqn:egt_eg}
		\lim_{t \to \infty}
		\E\left[g_t\left(\frac{X_{C_{1:m-1}}}{t}\right)
		\, \Big| \, X_u=t \right]
		=
		\E[g(A_{C_{1:m-1}})].
		\end{equation}
	The expectation on the left-hand side of~\eqref{eqn:egt_eg} is the integral in~\eqref{eq:int_y} while the right-hand side of~\eqref{eqn:egt_eg} is equal to
		\[
		\E[f(A_{u,C_{1:m-1}}, A_{us}Z_{s,C_m\setminus s})
		\ind(A_{us}\geq \delta)]
		=
		\E[f(A_{u,V\setminus u}) \, \ind(A_{us}\geq \delta)].
		\]
		Thus we have shown that~\eqref{eq:larger} converges to $0$ as $t\rightarrow\infty$, for any continuity point $\delta$ of $A_{us}$.
	
	\paragraph{\mdseries\itshape Analysis of~\eqref{eq:smaller}.}
	As $f$ is a function with range $[0, 1]$ we have
	\[
	0\leq \E[f(X_{V\setminus u}/t)\ind( X_{s}/t<\delta)\mid X_u=t]
	\leq
	\P[X_{s}/t<\delta\mid X_u=t]
	\]
	as well as
	\[
	0\leq \E[f(A_{u,V\setminus u})\ind(A_{us}<\delta)]
	\leq 
	\P[A_{us}<\delta] .
	\]
	By the triangle inequality and the two inequalities above,~\eqref{eq:smaller} is bounded from above by
	\begin{equation} \label{eq:pxpa}
	\P[X_{s}/t<\delta\mid X_u=t]
	+
	\P[A_{us}<\delta] .
	\end{equation}
	By the induction hypothesis 
	\[
	\lim_{t\rightarrow\infty}\P[X_{s}/t<\delta\mid X_u=t]
	=
	\P[A_{us}<\delta],
	\]
	and~\eqref{eq:pxpa} converges to $2 \P[A_{us}<\delta]$, which goes to 0 as $\delta\downarrow 0$ in case $\P(A_{us}=0)=0$.
	
	Suppose $\P(A_{us}=0)>0$. In this step we will need Assumption~\ref{ass:zero}. By the induction hypothesis, we have $A_{us}=\prod_{(a,b)\in \pth{u}{s}}Z_{ab}$ and the variables $Z_{ab}$ are independent. Hence 
	\begin{align*}
	  \P(A_{us}=0)=\P\Big(\min_{(a,b)\in\pth{u}{s}}Z_{ab}=0\Big)
	=
	1-\prod_{(a,b)\in \pth{u}{s}}\P(Z_{ab}>0).  
	\end{align*}
	If for any $(a,b)\in \pth{u}{s}$ we have $\P(Z_{ab}=0)>0$ then $P(Z_{ab}>0)<1$ and hence $\P(A_{us}=0)>0$. Therefore the assumption applies when the marginal distribution $\nu_{C,a}^b(\{0\})$ is positive. 

	Then by adding and subtracting terms and using the triangle inequality, we have the following upper bound for the term in~\eqref{eq:smaller}:
	\begin{align}
	&\left|\E\left[f\left(\frac{X_{C_{1:m-1}}}{t},\frac{X_{C_m\setminus s}}{t} \right)
	\ind\left(\frac{X_{s}}{t}<\delta\right)\Big| X_u=t\right]
	-\E\left[f\left(\frac{X_{C_{1:m-1}}}{t},0\right)\ind\left(\frac{X_{s}}{t}<\delta\right)\Big| X_u=t\right] \right|
	\label{eq:xx0}
	\\&+
	\left|\E\left[f(A_{u,C_{1:m-1}},A_{u,C_m\setminus s} )
	\ind(A_{u,s}<\delta)\right]
	-
	\E\left[f(A_{u,C_{1:m-1}},0)
	\ind(A_{u,s}<\delta)\right]\right|
	\label{eq:aa0}
	\\&+
	\left|\E\left[f\left(\frac{X_{C_{1:m-1}}}{t},0\right)
	\ind\left(\frac{X_{s}}{t}<\delta\right)\Big| X_u=t\right]
	-
	\E\left[f(A_{u,C_{1:m-1}},0)
	\ind(A_{u,s}<\delta)\right]\right|.
	\label{eq:x0a0}
	\end{align}
	We treat each of the three terms in turn.
	
		
Equation~\eqref{eq:x0a0} converges to $0$ by the induction hypothesis; note again that the set of discontinuities of the integrand receives zero probability in the limit. 
	
	Next we look at expression~\eqref{eq:xx0}. From the assumptions of $f$, namely that it ranges in $[0,1]$ and that $|f(x)-f(y)|\leq L \|x-y\|_1$ for some constant $L>0$, where $\|z\|_1 = \sum_j |z_j|$ for a Euclidean vector $z$, the term in~\eqref{eq:xx0} is bounded by 
	\begin{multline*}
	    \E\left[
	\left|f\left(\frac{X_{C_{1:m-1}}}{t},\frac{X_{C_m\setminus s}}{t} \right)
	-
	f\left(\frac{X_{C_{1:m-1}}}{t},0\right)\right|
	\ind\left(\frac{X_{s}}{t}<\delta\right)
	\, \Bigg| \, X_u=t\right]
	\\ \leq
	\E\left[\ind\left(X_s/t<\delta\right)
	\min\left(1,L \|X_{C_{m \setminus s}}/t\|_1 \right)
	\, \mid \, X_u=t\right].
	\end{multline*}
	We need to show that the upper bound converges to $0$ as $t\rightarrow\infty$. 
	Because the variables in $C_m\setminus s$ are independent of $X_u$ conditionally on $X_s$, the previous integral is equal to
	\begin{equation}
	\label{eq:int_L}
	\int_{[0,\delta]} 
	\E\left[
	\min\left(1,L\|X_{C_m \setminus s}/t\|_1\right)
	\, \mid \,
	X_{s} / t = x_{s} 
	\right]
	\P\left(X_{s}/{t}\in \mathrm{d}x_{s} \mid X_u=t \right).	
	\end{equation}
	For $\eta >0$, the inner expectation is equal to 
	\begin{align}
	&\E\left[\min\left(1,L\|X_{C_m \setminus s}/t\|_1\right) \,
	\ind \left\{ \forall v \in C_m \setminus s : X_v/t \le \eta \right\}
	\mid X_{s}/t = x_{s} \right]
	\label{eq:smaller_eta}
	\\& +
	\E\left[\min\left(1,L\|X_{C_m \setminus s}/t\|_1 \right) \,
	\ind \left\{ \exists v \in C_m\setminus s : X_v / t > \eta \right\}
	\mid X_{s} / t = x_{s}\right].
	\label{eq:larger_eta} 
	\end{align}
	The integrand in~\eqref{eq:smaller_eta} is either zero because of the indicator function or, if the indicator is one, it is bounded by $L \left|C_m \setminus s\right| \eta$.
	The expression in~\eqref{eq:larger_eta} is clearly smaller than or equal to
	$
	\P\left(
		\exists v\in C_m\setminus s : X_v / t > \eta \mid X_s/t = x_s
	\right)$.
	Going back to the integral in~\eqref{eq:int_L} we can thus bound it by
	\begin{align} \label{eq:delta_f}
	&\int_{[0,\delta]}
	\left[
		L \left|C_m\setminus s\right|\eta
		+
		\P\left(
			\exists v \in C_m \setminus s : X_v/t > \eta \mid X_s/t = x_s
		\right)
	\right]
	\P\left( X_s/t \in \mathrm{d}x_s \mid X_u=t\right).
	\end{align}
	Consider the supremum of the probability in the integrand over the values $x_{s}\in [0,\delta]$ to bound the integral further.
	Hence~\eqref{eq:delta_f} is smaller than or equal to 
	\[
	L \left| C_m\setminus s \right| \eta
	+
	\sup_{x_{s}\in [0,\delta]}
	\P \left( \exists v \in C_m \setminus s : X_v/t > \eta \mid X_s/t = x_s \right).
	\]
	Using Assumption~\ref{ass:zero} and the fact that $\eta$ can be chosen arbitrarily small we conclude that~\eqref{eq:xx0} converges to $0$ as $t\rightarrow\infty$.
	
	Finally we look at the term in~\eqref{eq:aa0}. As $f$ has range contained in $[0, 1]$ and is Lipschitz continuous, the expression in~\eqref{eq:aa0} is smaller than or equal to
	\begin{align} \label{eq:ALA}
	&    \E\Big[
	\ind(A_{us}<\delta)
	\min\Big(1, \, {L\textstyle\sum_{v\in C_m\setminus s}A_{uv}}\Big)\Big].
	\end{align}
	From $(A_{uv}, v\in C_m\setminus s)=A_{u,C_m\setminus s}=A_{us}Z_{s,C_m\setminus s}$ we can write~\eqref{eq:ALA} as
	\begin{align*}
	&    \E\Big[
	\ind(A_{us}<\delta)
	\min\Big(1,\textstyle{LA_{us}\sum_{v\in C_m\setminus s}Z_{sv}}\Big)\Big].
	\end{align*}
	The random variable inside the expectation is bounded by $1$ for any value of $\delta>0$ and it converges to $0$ as $\delta\downarrow 0$. By the bounded convergence theorem, the expectation in~\eqref{eq:ALA} converges to $0$ as $\delta \downarrow 0$.    
\qed

\subsection{Proof of Proposition~\ref{prop:nustdf}}

	The quantile function of the unit-Fréchet distribution is $u \mapsto -1/\ln(u)$ for $0 < u < 1$. In view of Sklar's theorem and the identity~\eqref{eqn:stdf}, the copula, $K$, of $G$ is
	\[
	K(u) 
	= G(-1/\ln u_1, \ldots, -1/\ln u_d) 
	= \exp \bigl( - \ell(- \ln u_1,\ldots,-\ln u_d) \bigr),
	\qquad u \in (0, 1)^d.
	\]
	It follows that the partial derivative $\dot{K}_1$ of $K$ with respect to its first argument exists, is continuous on $(0, 1)^2$ and is given by
	\[
	\dot{K}_1(u)
	= \frac{K(u)}{u_1} \, \dot{\ell}_1(-\ln u_1,\ldots,-\ln u_d),
	\]
	for $u \in (0, 1)^d$. The stdf is homogeneous: for $t > 0$ and $x \in (0, \infty)^d$, we have
	\[
	\ell(tx_1,\ldots,tx_d) = t \, \ell(x_1,\ldots,x_d).
	\]
	Taking the partial derivative with respect to $x_1$ on both sides and simplifying yields the identity
	\[
	\dot{\ell}_1(tx_1,\ldots,tx_d) = \dot{\ell}_1(x_1,\ldots,x_d).
	\]
	Let $F(x) = \exp(-1/x)$, for $x > 0$, denote the unit-Fréchet cumulative distribution function. Note that $-\ln F(x) = 1/x$ for $x > 0$. 
	For $t > 0$ and $x = (x_2,\ldots,x_d)\in (0, \infty)^{d-1}$, we find
	\begin{align*}
		\P(\forall j \ge 2, X_j \le t x_j \mid X_1 = t)
		&= \dot{K}_1 \bigl( F(t),F(tx_2),\ldots,F(tx_d) \bigr) \\
		&= \frac{K \bigl( F(t),F(tx_2),\ldots,F(tx_d) \bigr)}{F(t)} \,
		\dot{\ell}_1 (1/t, 1/(tx_2), \ldots, 1/(tx_d)) \\
		&= \frac{K \bigl( F(t),F(tx_2),\ldots,F(tx_d) \bigr)}{F(t)} \,
		\dot{\ell}_1 (1, 1/x_2, \ldots, 1/x_d).
	\end{align*}
	As $t \to \infty$, the first factor on the right-hand side tends to one, whence
	\[
	\lim_{t \to \infty}
	\P(\forall j \ge 2, X_j \le t x_j \mid X_1 = t)
	= \dot{\ell}_1 (1, 1/x_2, \ldots, 1/x_d).
	\]
	
	To show that the right-hand side of the previous equation is indeed the cumulative distribution function of a $(d-1)$-variate probability measure on Euclidean space, it is sufficient to show that, for every $j \in \{2, \ldots, d\}$, the family of conditional distributions $(X_j/t \mid X_1 = t)$ as $t$ ranges over $[t_0, \infty)$ for some large $t_0 > 0$ is uniformly tight.
	Indeed, the family of joint conditional distributions $\left((X_2,\ldots,X_d)/t \mid X_1 = t\right)$ for $t \in [t_0, \infty)$ is then uniformly tight as well, and by Prohorov's theorem \citep[Theorem~2.4]{vaart1998asymptotic}, we can find a sequence $t_n \to \infty$ such that the joint conditional distributions $(X/t_n \mid X_1 = t_n)$ converge weakly as $n \to \infty$, the limiting cumulative distribution function then necessarily being equal to the one stated above.
	It suffices to consider the case $d = j = 2$. By the first part of the proof above,
	\[
	\lim_{t \to \infty} \P(X_2/t > x_2 \mid X_1 = t) 
	= 1 - \dot{\ell}_1(1, 1/x_2).
	\]
	Since $\ell : [0, \infty)^2 \to [0, \infty)$ is convex, the functions $y_1 \mapsto \ell(y_1, y_2)$ depend continuously on the parameter $y_2 \ge 0$. Since they are also convex, Attouch's theorem \citep[Theorem~12.35]{rockafellar1998variational} implies that their derivatives depend continuously in $y_2$ as well, at least in points $y_1$ where $y_1 \mapsto \ell(y_1, y_2)$ is continuously differentiable. But since $\ell(y_1, 0) = y_1$, we find that $\dot{\ell}_1(1, 1/x_2) \to \dot{\ell}_1(1, 0) = 1$ as $x_2 \to \infty$. For any $\epsilon > 0$, we can thus find $x_2(\epsilon) > 0$ such that $1 - \dot{\ell}_1(1, 1/x_2(\epsilon)) < \epsilon / 2$ and then we can find $t(\epsilon) > 0$ such that $\P(X_2 / t > x_2(\epsilon) \mid X_1 = t) < \epsilon/2 + 1 - \dot{\ell}_1(1,1/x_2(\epsilon)) < \epsilon$ for all $t > t(\epsilon)$. The uniform tightness follows.
%
%
\qed

\section{Proofs for Section~\ref{sec:hr}}
\label{app:hr}

\subsection{Proof of Proposition~\ref{prop:log-excess}}

		By 
		Assumption~\ref{ass:HRMf}, the random vector $X$ satisfies Assumption~\ref{ass:nu} and it is Markov with respect to the graph $\G$.
		Assumption~\ref{ass:zero} is void (i.e., there is nothing to check), since, for each edge $(i, j) \in E$, the limiting distribution of $X_j/t \mid X_i = t$ as $t \to \infty$ is log-normal by \citet[Example~4]{segers2020one} and Example~\ref{exampleHR} here and therefore does not have an atom at zero.
		We can thus apply Theorem~\ref{prop_main} to conclude that $(X_v/t, v \in V \setminus u \mid X_u = t)$ converges weakly as $t \to \infty$.
		By the continuous mapping theorem, the same then holds true for $(\ln(X_v/t), v \in V \setminus u \mid X_u = t)$.
		It remains to calculate the limit distribution. 
		
		
\paragraph{\mdseries\itshape Calculating the limit in Theorem~\ref{prop_main}.}
	By example~\ref{exampleHR} we have, as $t \to \infty$,  
		\begin{equation} \label{eqn:log_conv}
			\bigl( \ln X_v - \ln X_s, \, v \in C \setminus s \mid X_s > t \bigr)
			\dto
			\mathcal{N}_{|C\setminus s|} \bigl(
			\mu_{C,s}(\Delta_C), \Psi_{C,s}(\Delta_C)
			\bigr),
		\end{equation}
		where the mean vector is
		\begin{align} \label{eqn:mulnX}
			\bigl(\mu_{C,s}(\Delta)\bigr)_v  = -2\delta_{sv}^2, \qquad v \in C \setminus s,
		\end{align}		
		and the covariance matrix $\Psi_{C,s}(\Delta)$ is as in~\eqref{eq:Psi_mat}.
		It follows that if the random vector $Z_{s,C \setminus s}$ has law $\nu_{C,s}$, then the distribution of $(\ln Z_{sv}, v \in C \setminus s)$ is equal to the limit in~\eqref{eqn:log_conv}. In particular, $\nu_{C,s}$ is multivariate log-normal.
	
	For fixed $u \in V$, we will identify the limit $A_{u, V \setminus u}$ in Theorem~\ref{prop_main}. 
	Let $Z = (Z_{s,C \setminus s}, C \in \mathcal{C})$ with $Z_{s,C} = (Z_{sv}, v \in C \setminus s)$ be the random vector constructed in Definition~\ref{def:increments} by concatenating independent log-normal random vectors with distributions $\nu_{C, s}$.
	In this concatenation, recall that $s \in C$ and that either $s$ is equal to $u$ or $s$ separates $u$ and $C \setminus s$.
	We can write $Z = (Z_e, e \in E_u)$ where the $E_u$ is the set of edges $e \in E$ that point away from $u$: for $e = (s, v) \in E_u$, either $s$ is equal to $u$ or $s$ separates $u$ and $v$.
	By construction, the distribution of $Z$ is multivariate log-normal too. By~\eqref{eqn:mulnX}, we have $\E[\ln Z_e] = -2\delta_e^2$ where $e = (s, v) \in E_u$. The covariance matrix of $(\ln Z_e, e \in E_u)$ has a block structure: for edges $e, f \in E_u$, the variables $\ln Z_e$ and $\ln Z_f$ are uncorrelated (and thus independent) if $e$ and $f$ belong to different cliques, while if they belong to the same clique, i.e., if $e = (s, i)$ and $f = (s, j)$ with $i,j,s \in C$ for some $C \in \mathcal{C}$, then, by~\eqref{eq:Psi_mat}, we have
	\begin{equation}
	\label{eq:covZeZf}
		\cov \left( \ln Z_e, \ln Z_f \right)
		= 2(\delta_{si}^2 + \delta_{sj}^2 - \delta_{ij}^2)
	\end{equation}
	
	By Theorem~\ref{prop_main}, we can express the limit $A_{u, V \setminus u}$ of $(X_v / t, v \in V \setminus u \mid X_u = t)$ as $t \to \infty$ in terms of $Z$: we have
	\begin{equation} \label{eqn:lAsumlZ}
		\ln A_{uv}
		= \ln \left(\prod_{e\in \pth{u}{v}}Z_e\right)
		= \sum_{e\in \pth{u}{v}}\ln Z_e,
		\qquad v \in V \setminus u.
	\end{equation} 
	The distribution of $(\ln A_{uv}, v \in V \setminus u)$ is thus multivariate Gaussian, being the one of a linear transformation of the multivariate Gaussian random vector $(\ln Z_e, e \in E_u)$.
	The expectation of $\ln A_{uv}$ is readily obtained from~\eqref{eqn:lAsumlZ}:
	\[
		\E[\ln A_{uv}] 
		= \sum_{e \in \pth{u}{v}} \E[\ln Z_e] 
		= \sum_{e \in \pth{u}{v}} (-2\delta_e^2) 
		= - 2p_{uv}, \qquad v \in V \setminus u,
	\] 
	which coincides with the element $v$ of the vector $\mu_u(\Delta)$ in~\eqref{eqn:muVu}.
	It remains to show that the covariance matrix of $(\ln A_{uv}, v \in V \setminus u)$ is $\Sigma_u(\Delta)$ in~\eqref{eqn:sigmaVu}.
	
	\paragraph{\mdseries\itshape Calculating $\Sigma_u(\Delta)$.}
	Let $i, j \in V \setminus u$. By~\eqref{eqn:lAsumlZ} and the bilinearity of the covariance operator, we have
	\begin{align*}
		\cov \left( \ln A_{ui}, \ln A_{uj} \right)
		= \sum_{e \in \pth{u}{i}} \sum_{f \in \pth{u}{j}} 
		\cov \left( \ln Z_e, \ln Z_f \right).
	\end{align*}
	Each of the paths $\pth{u}{i}$ and $\pth{u}{j}$ has at most a single edge in a given clique $C \in \mathcal{C}$; otherwise, they would not be the shortest paths from $u$ to $i$ and $j$, respectively.
	Let the node $a \in V$ be such that $\pth{u}{i} \cap \pth{u}{j} = \pth{u}{a}$. It could be that $a = u$, in which case the intersection is empty. Now we need to consider three cases.
	\begin{enumerate}[1.]
	\item 
	If $a = i$, i.e., if $i$ lies on the path from $u$ to $j$, then the random variables $\ln Z_f$ for $f \in \pth{i}{j}$ are uncorrelated with the variables $\ln Z_e$ for $e \in \pth{u}{i}$. By~\eqref{eq:covZeZf}, the covariance becomes
	\begin{align*}
		\cov \left( \ln A_{ui}, \ln A_{uj} \right)
		&= \sum_{e \in \pth{u}{i}} \var \left(\ln Z_e\right) \\
		&= \sum_{e \in \pth{u}{i}} 4 \delta_e^2 
		= 4 p_{ui} = 2 \left(p_{ui} + p_{uj} - p_{ij}\right),
	\end{align*}
	since $p_{ui} + p_{ij} = p_{uj}$, the path from $u$ to $j$ passing by $i$. This case includes the one where $i = j$, since then $\pth{i}{j}$ is empty and thus $p_{ij} = 0$.
	
	\item
	If $a = j$, the argument is the same as in the previous case.
	
	\item
	Suppose $a$ is different from both $i$ and $j$. Let $e_a$ and $f_a$ be the first edges of the paths $\pth{a}{i}$ and $\pth{a}{j}$, respectively. These two edges may or may not belong to the same clique. All other edges on $\pth{a}{i}$ and $\pth{a}{j}$, however, must belong to different cliques. It follows that
	\begin{align*}
	\cov \left( \ln A_{ui}, \ln A_{uj} \right)
	&= \sum_{e \in \pth{u}{a}} \var \left( \ln Z_e \right)
	+ \cov \left( \ln Z_{e_a}, \ln Z_{f_a} \right) \\
	&= 4 p_{ua} + \cov \left( \ln Z_{e_a}, \ln Z_{f_a} \right).
	\end{align*}
	Now we need to distinguish between two further sub-cases.
	\begin{enumerate}[({3}.a)]
	\item Suppose $e_a$ and $f_a$ do not belong to the same clique. Then the covariance between $\ln Z_{e_a}$ and $\ln Z_{f_a}$ is zero, so that
	\begin{align*}
		\cov \left( \ln A_{ui}, \ln A_{uj} \right)
		= 4 p_{ua} 
		&= 2 \left( 
			\left(p_{ui} - p_{ai}\right) + \left(p_{uj} - p_{aj}\right) 
		\right) \\
		&= 2 \left( p_{ui} + p_{uj} - \left( p_{ai} + p_{aj} \right) \right) \\
		&= 2 \left( p_{ui} + p_{uj} - p_{ij} \right),
	\end{align*}
	since the shortest path between $i$ and $j$ passes through $a$.
	\item Suppose $e_a$ and $f_a$ belong to the same clique; see Figure~\ref{fig:proof:log-excess}. Writing $e_a = (a, k)$ and $f_a = (a, l)$, we find, in view of~\eqref{eq:covZeZf},
	\begin{align*}
		\cov \left( \ln A_{ui}, \ln A_{uj} \right)
		&= 4 p_{ua} + 2 \left( \delta_{ak}^2 + \delta_{al}^2 - \delta_{kl}^2\right) \\
		&= 2 \left( \left(p_{ua} + \delta_{ak}^2\right) + \left(p_{ua} + \delta_{al}^2\right) - \delta_{kl}^2 \right) \\
		&= 2 \left( p_{uk} + p_{ul} - \delta_{kl}^2 \right) \\
		&= 2 \left( 
			\left(p_{ui} - p_{ki}\right) + 
			\left(p_{uj} - p_{lj}\right) - 
			\delta_{kl}^2 \right) \\
		&= 2 \left( p_{ui} + p_{uj} - \left( p_{ki} + p_{lj} + \delta_{kl}^2 \right) \right) \\
		&= 2 \left( p_{ui} + p_{uj} - p_{ij} \right),
	\end{align*}
	since the shortest path between $i$ and $j$ passes by $k$ and $l$.
	\end{enumerate}
	\end{enumerate}
\begin{figure}
\begin{center}
	\begin{tikzpicture}
		\node[hollow node] (a) at (0,0)  {$a$};
		\node[hollow node] (k) at (1.5,1)  {$k$};
		\node[hollow node] (l) at (1.5,-1)  {$l$};
		\node[] (e1) at (3,1)  {$\cdots$};
		\node[] (e2) at (3,-1)  {$\cdots$};
		\node[] (e3) at (-1.5,0)  {$\cdots$};
		\node[] (i) at (4.5,1)  {$i$};
		\node[] (j) at (4.5,-1)  {$j$};
		\node[] (u) at (-3,0)  {$u$};
		\node[] (c) at (0.15,1.6)  {$C$};
		\path (a) edge (k);
		\path (a) -- (k) node [midway,auto=left] {$e_a$};
		\path (a) edge (l);
		\path (a) -- (l) node [midway,auto=right] {$f_a$};
		\path (k) edge (l);
		\path (k) edge (e1);
		\path (l) edge (e2);
		\path (a) edge (e3);
		\path (e1) edge (i);
		\path (e2) edge (j);
		\path (e3) edge (u);
		\begin{scope}[very thick,dashed]
			\draw[] (0.9,0) ellipse (38pt and 46pt);
		\end{scope}
	\end{tikzpicture}
\end{center}
\caption{\label{fig:proof:log-excess} Calculation of $\cov \left( \ln A_{ui}, \ln A_{uj} \right)$ in the proof of Proposition~\ref{prop:log-excess}. The paths from $u$ to $i$ and from $u$ to $j$ have the path from $u$ to $a$ in common. On these two paths, the edges right after node~$a$ are $e_a = (a, k)$ and $f_a = (a, l)$, respectively. The picture considers the case (3.b) where the three nodes $a, k, l$ belong to the same clique, say $C$.}
\end{figure}
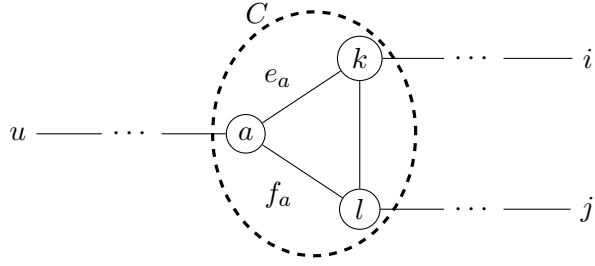
We conclude that the covariance matrix of $(\ln A_{uv}, v \in V \setminus u)$ is indeed $\Sigma_u(\Delta)$ in~\eqref{eqn:sigmaVu}.

\paragraph{\mdseries\itshape Positive definiteness of $\Sigma_u(\Delta)$.}

Being a covariance matrix, $\Sigma_u(\Delta)$ is positive semi-definite. We need to show it is invertible.
The linear transformation in~\eqref{eqn:lAsumlZ} can be inverted to give
\[
\ln Z_e = \ln A_{ub} - \ln A_{ua}
\]
for an edge $e = (a, b)$ in $E_u$; note that either $a = u$, in which case $A_{ua} = 1$ and thus $\ln A_{ua} = 0$, or $a$ lies on the path from $u$ to $b$.
Also, each edge $e \in E_u$ is uniquely identified by its endpoint $v$ in $V \setminus u$; let $e(v)$ be the unique edge in $E_u$ with endpoint $v$.
It follows that, as column vectors, the random vectors $(\ln Z_{e(v)}, v \in V \setminus u)$ and $(\ln A_{uv}, v \in V \setminus u)$ are related by
\[
	\left( \ln A_{uv}, v \in V \setminus u \right)
	= M_u \left( \ln Z_{e(v)}, v \in V \setminus u \right), 
\]
where $M_u$ is a $(|V|-1) \times (|V|-1)$ matrix indexed by $(v, w) \in (V \setminus u)^2$ whose inverse is
\[
	(M_u^{-1})_{vw} =
	\begin{cases}
		\phantom{-}1 & \text{if $w = v$,} \\
		-1 & \text{if $(w, v) \in E_u$,} \\
		\phantom{-}0 & \text{otherwise.}
	\end{cases}
\]
The covariance matrix of $(\ln A_{uv}, v \in V \setminus u)$ is thus
\[
	\Sigma_u(\Delta) = M_u \, \Sigma_u^Z(\Delta) \, M_u^\top
\]
where $\Sigma_u^Z(\Delta)$ is the (block-diagonal) covariance matrix of $(\ln Z_{e(v)}, v \in V \setminus u)$. 
The blocks in $\Sigma_u^Z$ are given by~\eqref{eq:covZeZf} and are positive definite and thus invertible by the assumption that each parameter matrix $\Delta_C$ is conditionally negative definite.
As a consequence, $\Sigma_u^Z(\Delta)$ is invertible too.
Writing $\Theta_u^Z(\Delta) = \bigl( \Sigma_u^Z(\Delta) \bigr)^{-1}$, we find that $\Sigma_u(\Delta)$ is invertible as well with inverse
\begin{equation}
\label{eqn:ThetauDelta}
	\Theta_u(\Delta) = 
	(M_u^{-1})^\top \, \Theta_u^Z(\Delta) \, M_u^{-1}.
\end{equation}

\paragraph{\mdseries\itshape $P(\Delta)$ is conditionally negative definite.}

Clearly, $P(\Delta)$ is symmetric and has zero diagonal. For any non-zero vector $a \in \reals^V$, we have, since $\Sigma_u(\Delta)$ is positive definite,
\begin{align*}
	0 &< a^\top \Sigma_u(\Delta) a \\
	&= 2 \sum_{i \in V} \sum_{j \in V} a_i \left( p_{ui} + p_{uj} - p_{ij} \right) a_j \\
	&= 2 \sum_{i \in V} a_i p_{ui} \sum_{j \in V} a_j + 2 \sum_{i \in V} a_i \sum_{j \in V} p_{uj} - 2 \sum_{i \in V} \sum_{j \in V} a_i p_{ij} u_j.
\end{align*}
If $\sum_{i \in V} a_i = 0$, the first two terms on the right-hand side vanish.
The last term is $-2 a^\top P(\Delta) a$.
We conclude that $P(\Delta)$ is conditionally negative definite, as required.
\qed

\subsection{Proof of Proposition~\ref{prop_HRattractor}}

Let $H_{P(\Delta)}$ be the $|V|$-variate max-stable Hüsler--Reiss distribution in~\eqref{eq:HRDelta} with parameter matrix $P(\Delta)$ in~\eqref{eq:pij}. From \eqref{eq:A2stdf} we have 
\[
-\ln H_{P(\Delta)}(x_v, v\in V)
= \ell(1/x_v, v\in V)
= \E\left[\max_{v \in V} x_v^{-1} A_{uv}\right].
\]
By the maximum--minimums identity, we have 
\[
	\E\left[\max_{v \in V} x_v^{-1} A_{uv}\right]
	=\sum_{i=1}^{|V|} (-1)^{i-1} \sum_{W \subseteq V: |W|=i} 
	\E\left[\min_{v \in V} x_v^{-1}A_{uv}\right]
\]

If $W$ is a singleton, $W = \{u\}$, then the expectation is simply $x_u^{-1}$, whereas if $W$ has more than one element, we write the expectation as the integral of the tail probability:
\begin{align*}
\E \left[ \min_{v \in W} x_v^{-1} A_{uv} \right]
&= \int_0^\infty 
\P \left[ \forall v \in W : x_v^{-1} A_{uv} > y \right]
\diff y \\
&= \int_0^{x_u^{-1}} 
\P \left[ \forall v \in W \setminus u : A_{uv} > x_v y \right] 
\diff y \\
&= \int_{\ln x_u}^{\infty} \P \left[
\forall v \in W \setminus u : \ln A_{uv} > \ln(x_v)-z 
\right] \, e^{-z} \, \diff z.
\end{align*}
If $W = \{u\}$, we interpret the probability inside the integral as equal to one, so that the integral formula is valid for any non-empty $W \subseteq V$.
Combining things, we find that
\begin{multline*}
-\ln H_{P(\Delta)}(x_v, v\in V)\\
= \sum_{i=1}^{|V|} (-1)^{i-1} 
\sum_{W \subseteq V: |W| = i}
\int_{\ln x_u}^{\infty}	\P \left[
\forall v \in W \setminus u : \ln A_{uv} > \ln (x_v) - z 
\right] e^{-z} \, \diff z.
\end{multline*}
Recall that the distribution of $(\ln A_{uv}, v \in V \setminus u)$ is multivariate normal with mean vector $\mu_u(\Delta)$ and covariance matrix $\Sigma_u(\Delta)$ in~\eqref{eqn:muVu} and~\eqref{eqn:sigmaVu}, respectively, hence the expression of $H_{P(\Delta)}$ in~\eqref{eq:HRDelta}.
\qed

\subsection{Proof of Proposition~\ref{prop:extr_gm}}

By Proposition~\ref{prop_HRattractor}, the Markov field $X$ in Assumption~\ref{ass:HRMf} satisfies~\eqref{eqn:FtG} with $F$ its joint cumulative distribution function and $G = H_{P(\Delta)}$.
It follows that~\eqref{eqn:pareto} holds too, yielding the weak convergence of high-threshold excesses to a Pareto random vector $Y$ with distribution given in~\eqref{eqn:pareto_lim}.
It remains to show that $Y$ is an extremal graphical model with respect to the given graph $\G$ in the sense of Definition~2 in \citet{engelke2020graphical}.  

Proposition~\ref{prop:log-excess} implies
\[
	\left( \ln X_v - \ln X_u, v \in V \setminus u \mid X_u > t \right)
	\dto \mathcal{N}_{|V \setminus u|} \bigl( \mu_u(\Delta), \Sigma_u(\Delta) \bigr)
\]
as $t \to \infty$. The latter is the distribution of $(\ln Y_v - \ln Y_u, v \in V \setminus u \mid Y_u > 1)$ for the multivariate Pareto random vector $Y$ in~\eqref{eqn:pareto} associated to the max-stable Hüsler--Reiss distribution with parameter matrix $P(\Delta)$.

To show that $Y$ is an extremal graphical model with respect to the given block graph $\G$, we apply the criterion in Proposition~3 in \citet{engelke2020graphical}.
Let $\Theta_u(\Delta) = (\Sigma_u(\Delta))^{-1}$ be the precision matrix of the covariance matrix $\Sigma_u(\Delta)$ in~\eqref{eqn:sigmaVu}; see~\eqref{eqn:ThetauDelta}.
For $i, j \in V$ such that $i$ and $j$ are not connected, i.e., $(i, j)$ is not an edge, we need to show that there is $u \in V \setminus \{i, j\}$ such that
\[
	\bigl(\Theta_u(\Delta)\bigr)_{ij} = 0.
\]
Indeed, according to the cited proposition, the latter identity implies that
\[
	Y_i \perp_{\mathrm{e}} Y_j \mid Y_{\setminus \{i, j\}},
\]
the relation $\perp_{\mathrm{e}}$ being defined in Definition~1 in \citet{engelke2020graphical}, and thus that $Y$ is an extremal graphical model with respect to $\G$.

For two distinct and non-connected nodes $i, j \in V$, let $u \in V \setminus \{i, j\}$.
We will show that $(\Theta_u(\Delta))_{ij} = 0$.
We have
\begin{align*}
	\bigl(\Theta_u(\Delta)\bigr)_{ij}
	&= \sum_{a \in V \setminus u} \sum_{b \in V \setminus u}
	\bigl((M_u^{-1})^\top\bigr)_{ia} \bigl(\Theta_u^Z(\Delta)\bigr)_{ab} \bigl(M_u^{-1}\bigr)_{bj} \\
	&= \sum_{a \in V \setminus u} \sum_{b \in V \setminus u}
	\bigl(M_u^{-1}\bigr)_{ai} \bigl(\Theta_u^Z(\Delta)\bigr)_{ab} \bigl(M_u^{-1}\bigr)_{bj}.
\end{align*}
Now, $\bigl(M_u^{-1}\bigr)_{ai}$ and $\bigl(M_u^{-1}\bigr)_{bj}$ are non-zero only if $a = i$ or $(i, a) \in E_u$ together with $b = j$ or $(j, b) \in E_u$. In neither case can $a$ and $b$ belong to the same clique, because otherwise we would have found a cycle connecting the nodes $u, i, a, b, j$.
But if $a$ and $b$ belong to different cliques, then so do the edges $e(a)$ and $e(b)$ in $E_u$ with endpoints $a$ and $b$, and thus  $\bigl(\Theta_u^Z(\Delta)\bigr)_{ab} = 0$, since $\Sigma_u^Z(\Delta)$ and thus $\Theta_u^Z(\Delta)$ are block-diagonal.
\qed

\subsection{Proof of Proposition~\ref{prop:identif_hr}}
  
\paragraph{\mdseries\itshape Necessity.}
Let $v \in \bar{U}$ have clique degree $\cd(v)$ at most two.
We show that the full path sum matrix $P(\Delta)$ is not uniquely determined by the restricted path sum matrix $P(\Delta)_U$ and the graph $\G$. There are two cases: $\cd(v) = 1$ and $\cd(v) = 2$.

Suppose first that $\cd(v) = 1$.
Then $v$ belongs only to a single clique, say $C \in \mathcal{C}$.
For any $i, j \in U$, the shortest path $\pth{i}{j}$ does not pass through $v$.
Hence the edge weights $\delta_{vw}^2$ for $w \in C \setminus v$ do not show up in any path sum $p_{ij}$ appearing in $P(\Delta)_U$.
It follows that these edge weights can be chosen freely (subject to specifying a valid Hüsler--Reiss parameter matrix on $C$) without affecting $P(\Delta)_U$.

Suppose next that $\cd(v) = 2$.
Without loss of generality, assume $U = V \setminus v$; this only enlarges the number of visible path sums with respect to the initial problem.
We will show that the path sum sub-matrix $(p_{ab}(\Delta))_{a, b \in V \setminus v}$ does not determine the complete path sum matrix $P(\Delta)$.

By assumption, $v$ is included in two different cliques.
Let the set of nodes from one of them, excluding $v$, be $I$ and let the set of nodes from the other one, excluding $v$, be $J$.
The sets $I$ and $J$ are non-empty and disjoint.
We will show that the edge parameters $\delta_{vi}^2$ for $i \in I$ and $\delta_{vj}^2$ for $j \in J$ are not uniquely determined by the path sums $p_{ab}$ for $a, b \in V \setminus \{v\}$.

\begin{itemize}
\item
On the one hand, if the path $\pth{a}{b}$ does not pass by $v$, then the path sum $p_{ab}$ does not contain any of the edge parameters $\delta_{vi}^2$ or $\delta_{vj}^2$ as a summand.
\item
On the other hand, if the path $\pth{a}{b}$ passes through $v$, then for some $i \in I$ and $j \in J$ determined by $a$ and $b$ the path sum $p_{ab}$ contains the sum $\delta_{vi}^2 + \delta_{vj}^2$ as a summand.
However, sums of the latter form do not change if we decrease each $\delta_{vi}^2$ ($i \in I$) by some small quantity, say $\eta > 0$, and simultaneously increase each $\delta_{vj}^2$ ($j \in J$) by the same quantity, yielding $(\delta_{vi}^2 - \eta) + (\delta_{vj}^2 + \eta) = \delta_{vi}^2 + \delta_{vj}^2$.
\end{itemize}

\paragraph{\mdseries\itshape Sufficiency.}
Let every node in $\bar{U}$ have clique degree at least three.
Let $a \in \bar{U}$ and let $\delta^2_{ab}$ be the parameter attached to the edge $(a,b)\in E$, with $b \in V \setminus a$.
We will show that we can solve $\delta^2_{ab}$ from the observable path sums $p_{ij}$ for $i, j \in U$.

By assumption there are at least three cliques that are connected to $a$, say $I$, $J$, and $Y$. Without loss of generality, assume $b\in I$.
If $b\in U$ set $\bar{\imath}:=b$, while if $b\in \bar{U}$ walk away from $b$ up to the first node $\bar{\imath}$ in $U$ and this along the unique shortest path between $b$ and $\bar{\imath}$; note that $(a,b) \in \pth{a}{\bar{\imath}}$.
Apply a similar procedure to the cliques $J$ and $Y$:
choose a node $j\in J \setminus a$ (respectively $y\in Y \setminus a$) and if $j\in U$ ($y\in U$) set $\bar{\jmath}:=j$ ($\bar{y}:=y$), while if $j\in \bar{U}$ (respectively $y\in \bar{U}$) take the first node $\bar{\jmath}$ ($\bar{y}$) such that $(a,j)\in \pth{a}{\bar{\jmath}}$ [$(a,y)\in \pth{a}{\bar{y}}$].
Because the nodes $\bar{\imath}$, $\bar{\jmath}$ and $\bar{y}$ belong to $U$, the path sums $p_{\bar{\imath}\bar{\jmath}}$, $p_{\bar{\imath}\bar{y}}$, and $p_{\bar{y}\bar{\jmath}}$ are given. By construction, node $a$ lies on the unique shortest paths between the nodes $\bar{\imath}$, $\bar{\jmath}$ and $\bar{y}$; see also \citet[Theorem~1(b)]{behtoei2010a}.
It follows that
\begin{align*}
	p_{\bar{\imath}\bar{\jmath}} &= p_{a\bar{\imath}}+p_{a\bar{\jmath}}, \\
	p_{\bar{\imath}\bar{y}}  &= p_{a\bar{\imath}}+p_{a\bar{y}}, \\
	p_{\bar{y}\bar{\jmath}}  &= p_{a\bar{\jmath}}+p_{a\bar{y}}.
\end{align*}
These are three equations in three unknowns, which can be solved to give, among others, $p_{a\bar{\imath}} = \frac{1}{2} (p_{\bar{\imath}\bar{y}} + p_{\bar{\imath}\bar{\jmath}} - p_{\bar{y}\bar{\jmath}})$.
Now we distinguish between two cases, $b \in U$ and $b \in \bar{U}$.

\begin{itemize}
\item
If $b \in U$ then $\bar{\imath} = b$ and we have written $\delta^2_{ab} = p_{a\bar{\imath}}$ in terms of the given path sums.

\item
If $b\in \bar{U}$ we repeat the same procedure as above but starting from node $b$.
We keep the node $\bar{\imath}$, but the nodes $\bar{\jmath}$ and $\bar{y}$ may be different from those when starting from $a$.
After having written $p_{b\bar{\imath}}$ in terms of observable path sums, we can compute $\delta_{ab}^2 = p_{a\bar{\imath}} - p_{b\bar{\imath}}$. \qed
\end{itemize}

\end{document}